\documentclass[11pt]{article}
 \usepackage{geometry}                
\usepackage{graphicx,amssymb,amsthm, amsmath, amsopn,amscd}
\usepackage{epstopdf}
\theoremstyle{plain}
\newtheorem*{theorem*}{Theorem}
\newtheorem{theorem}{Theorem}
\newtheorem{proposition}{Proposition}[section]
\newtheorem{lemma}[proposition]{Lemma}

\newtheorem{corollary}[proposition]{Corollary}
\newcount\conjno
\def\conjecture#1{\global\advance\conjno by 1\bigskip\noindent{\bf Conjecture
\the\conjno.} {\it #1}\\}

\newtheorem*{definition*}{Definition}

\theoremstyle{definition}       
\newtheorem{definition}{Definition}[section]   
\newtheorem{remark}{Remark}[section]

\def\Cal{\mathcal}
\def\ie{{\it i.e. }}

\def\real{{\mathbb R}}
\def\r2{{\mathbb R}^2}

\def\cmplx{{\mathbb C}}
\def\intgr{{\mathbb Z}}
\def\circle{{\mathbb S^1}}

\def\cyl{{\Cal C}}      
\def\lattice{L}
\def\config{{\bf p}}
\def\tiling{{T}}
\def\tilingcov{\tilde{T}}
\def\rtmnd{{\cal RT}(M,N,D)}   
\def\rtmndk{{\cal RT}_{(M,N,D,K)}}   
\def\dmndk{{\cal DT}_{(M,N,D,K)}}  

\def\dd{\vec{d}}        
\def\uu{\vec{u}}        
\def\zz{\underline z}			
\def\diam{D}		
\def\ph{phyllotactic}
\def\PH{Phyllotactic}

\def\snow{$S$}
\def\snowm{S}
\def\qsnow{$\overline S$}
\def\qsnowm{\overline S}

\def\ie{{\it i.e.}\ }   
\def\eg{{\it e.g.}\ }
\def\resp{{\it resp.}}
\def\sameas{\Leftrightarrow}

\def\2vec#1#2{(#1, #2)}
\def\widevec#1{\overset{\hbox{\rightarrowfill}}{#1}} 
\def\ovv#1{\overline{#1}}

\begin{document}

\title{Rhombic Tilings and Primordia Fronts of Phyllotaxis}
\author{Pau Atela and Christophe Gol\'e\\ Department of
Mathematics\\Smith College } 
\date{}

\maketitle
\abstract{
We introduce and study properties of \ph~and rhombic tilings on the cylinder. These are discrete sets of points that generalize cylindrical lattices. Rhombic tilings appear as periodic orbits of a discrete dynamical system \snow~that models plant pattern formation by stacking disks of equal radius on the cylinder.  This system has the advantage of allowing several disks at the same level, and thus multi-jugate configurations.  We provide partial results toward proving that the attractor for \snow~is entirely composed of rhombic tilings and is a strongly normally attracting branched manifold and conjecture that this attractor persists topologically  in nearby systems. A key tool in understanding the geometry of tilings and the dynamics of \snow~is the concept of primordia front, which is a closed ring of tangent disks around the cylinder. We show how fronts determine the dynamics, including transitions of parastichy numbers, and might explain the Fibonacci number of petals often encountered in {\it compositae}.}
\section{Introduction} 
Phyllotaxis is the study of arrangements of plant organs.  These originate at the growing tip (apex meristem) of a plant as protuberances of cells, called primordia.
The geometric classification of phyllotactic patterns has often been reduced to that of cylindrical lattices, where the helices joining nearest primordia - called parastichies -  form two families winding in opposite directions. Counting parastichies in each family gives rise to the pair of parastichy numbers that are used to classify phyllotactic patterns. The striking phenomenon central to phyllotaxis is the predominance of pairs of successive Fibonacci numbers as parastichy numbers.

 \begin{figure}[h] 
   \centering
 \includegraphics[width=6in]{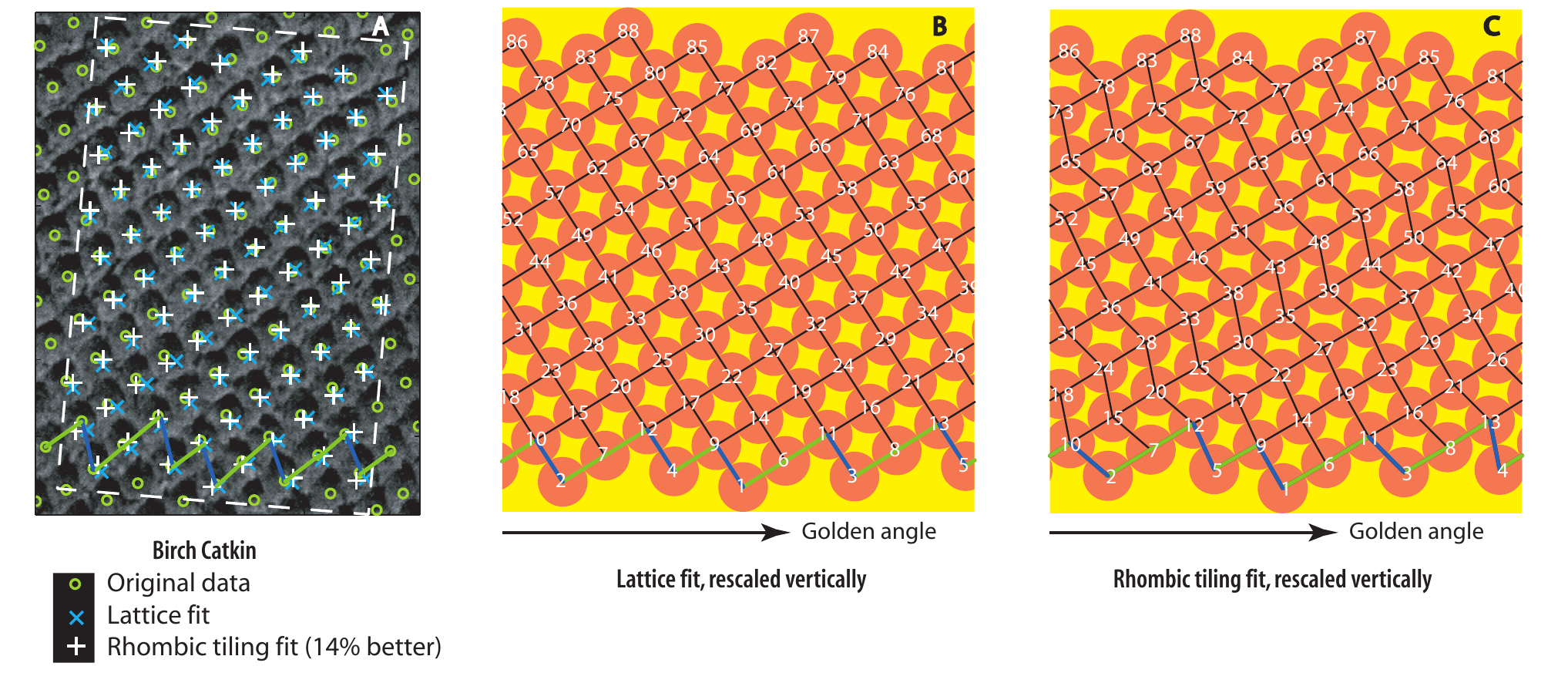} 
   \caption{\footnotesize {Each picture represents the unrolled surface of a cylinder. {\sf (A)}  The imprint of a Birch catkin rolled on clay. After a vertical compression counterbalancing anisotropic growth, we used a method of least squares with (nonlinear) constraints to fit lattices and rhombic tilings to this pattern. The results of this process are seen in the lattice in {\sf (B)}, the tiling in {\sf (C)}.  Note that there are 5 blue segments and 8 green in the fronts represented at the bottom of the pattern: they form the front parastichy numbers, and coincide with the number of parastichies of both figures.  One can think of  the tiling as a deformation of the lattice, obtained by rotating the segments of the front of the lattice. The tiling offers an improvement of the fit of more than 14 percent over the lattice, and it also accounts for the undulations of the parastichies. Note that, even though the parastichies are well defined in the rhombic tiling, the divergence angle between successive primordia (numbered according to height, and thus, presumably according to age) is widely erratic (although periodic). See for instance the differences of horizontal spacing between Primordia pairs 42, 43 and 43, 44.}}
   \label{fig:tilingfit}
\end{figure}

However, Fibonacci patterns and transitions among these are not the only ones observed in nature.   A very common transition can be seen on stems of sunflowers, for instance: after a few pairs of aligned leaves alternating at a $90^o$ angle  leaves suddenly grow in spirals yielding Fibonacci numbers.  In terms of parastichy numbers classification, the pattern with parastichy numbers $(2,2)$, (decussate), transitions to $(2,3)$. This transition is usually absent from the analysis of dynamical models of phyllotaxis, even when they can reproduce it. More generally, transitions to and from multijugate phyllotaxis, where parastichy numbers have a common divisor $k$, and  where $k$ organs appear at the same level (whorl), is not often discussed  (\cite{douadycouder}, Parts II \& III being a notable exception).  Part of the difficulty lied in the absence, in the literature, of a continuum of patterns encompassing lattices of all jugacies, and of more local geometric tools to follow the transitions as they unfold one primordium at a time.

We introduce the geometric concept of \ph~and rhombic tilings, which do encompass lattices of all jugacies, as well as patterns hitherto considered as transient. These tilings can be seen as deformations of cylindrical lattices. In contrast with lattices, they can account for the marked undulations of parastichies often observed in nature (Fig.  \ref{fig:tilingfit}).

We also reintroduce van Iterson's century old concept of ``zickzacklinie" \cite{vaniterson}, that we call here  primordia fronts (Fig.  \ref{fig:tilingfit}). These zig-zaging fronts and their parastichy numbers offer a practical and theoretical tool to understand not only the steady state tilings but also transitions from one to another, in a way that may be less confusing than the divergence angles often used in experiments (see e.g. Section \ref{subsec:fibonacci}). The concept of primordia front might also offer an explanation as to the statistical predominance of Fibonacci
numbers of ray petals in many asteracea (\cite{battjes}): the number of primordia in a front  is the sum of its Fibonacci (likely) parastichy numbers, hence itself a (likely) Fibonacci number.

We root the concepts of tilings and fronts within a simple discrete dynamical model that more or less explicitly exists since the 19th century (\cite{weisse}, \cite{vaniterson}, \cite{williams}, \cite{douady}). This system, that we call the Snow map after  \cite{douadycouder} and denote by \snow~, represents primordia formation as the stacking of disks on a cylinder, according to the simple rules: the new disk appears at the lowest level above the older ones, without overlap. As we fix the circumference of the cylinder, the diameter $D$ of the of the primordia is the fundamental parameter of this model.

Theorem \ref{thm:attractor}, brings together the geometry and dynamics of this paper, by showing that, for each parastichy number pair $M,N$  and for $D$ in an appropriate range, there exists a manifold of rhombic tilings, each of which is a periodic orbit of period $MN$ for \snow. This manifold is of dimension $M+N$, and contains the $M,N$-lattice of the fixed point bifurcation diagram for that $D$ (see Section \ref{subsec:fp.po}).  We also show that this manifold is a local attractor for \snow, and that the attraction occurs in \emph{finite time}.

We conjecture that the entire set of dynamically sustainable rhombic tilings forms a normally attracting set. This  should imply the persistence of an invariant set with comparable topology in nearby models (\cite{fenichel}) and would confer the Snow model and rhombic tilings some universality in phyllotaxis.

Although we do not study phyllotactic transitions in great detail here (see Sections \ref{subsec:fibonacci} and \ref{subsec:frontdyn}), we hope that this paper will serve as foundation for further research in that direction.  Later work will explore  the topological structure of the set of dynamically sustainable tilings, and of the dynamical transitions it allows, as well as  generalizations of these tilings to other geometries. Experimental applications of some of  the concepts discussed here, such as using fronts derived from plant data as initial conditions for growth modeling using a similar model, appeared in \cite{jpgr}. 

Recent experimental and modeling work points to the active transport of the hormone auxin \cite{auxinbern}, \cite{auxintraas},\cite{auxinmjolsness}, \cite{auxinprunsi} as the underlying mechanism of primordia formation, although some authors still advocate for a buckling explanation \cite{shipman}.  Although the type of models based on auxin transport should eventually  prove invaluable in testing the validity of proposed biological mechanisms, to date they can't easily and stably reproduce  Fibonacci phyllotaxis, and neither could they form the proper context for a geometrical explanation of its prominence. Our approach is grounded in the tradition of  dynamical/geometric models (\cite{weisse},\cite {vaniterson}, \cite{williams}, \cite{adler}, \cite{douadycouder}, \cite{douady}, \cite{leelevitov}, \cite{kunzthesis}, \cite{jns}), often based  on the botanical observations of Hofmeister \cite{hofmeister} and Snow \& Snow \cite{snow}. The model we study is also compatible with the general assumptions of  \cite{auxinbern} and  \cite{shipman}.  Our goal is to distill to their simplest and most rigorous form the geometric mechanisms that could be at play in Phyllotactic pattern formation. The concepts we develop are general enough that they may adapt to other situations, such as the assembly of the HIV-1 CA protein \cite{HIVnature}.

To motivate this otherwise rather theoretical paper, we start in Section \ref{section:numexp} by reporting on some numerical experiments, showing how phenomena encountered by iterating \snow~on a computer naturally lead to rhombic tilings and primordia fronts. We then review the classical geometry of the  cylindrical lattices and of their parastichies (Section \ref{sec:classicgeom}). In  Section \ref{subsec:parentsEtc}, we  establish the notion, for general configurations of primordia, of chains and fronts of primordia as sets of tangent primordia encircling the meristem. The parastichy numbers of chains and fronts are just the number of up and down segments as one travels around the chain.  The definition of \ph~and rhombic tiling as cyclic sums of up and down vectors follows in Section \ref{subsec:tilings}, followed by the analysis of their periodicity and properties of parastichies. In Section \ref{subsec:parastnum}, we show the equality of parastichy numbers of a tiling and of any of its fronts - thus validating the usage of the front parastichy numbers. 

We then give a rigorous definition of the Snow map \snow, followed by a study of its domain of differentiability (Section \ref{sec:snow}).
Section \ref{sec:dyngeom} brings the dynamics of \snow~and the geometry of fronts and tilings together. In Section \ref{subsec:frontdyn}, we show that the top primordia front of a configuration determines its dynamical future, and that changes in parastichy numbers can be simply read from the number of sides of the polygonal tile between a new primordium and the top front. We show that the fixed points of the map \qsnow~induced on the shape space of configurations are the same rhombic cylindrical lattices as in the Hofmeister map of \cite{jns} and conjecture  that periodic orbits all form rhombic tilings (Section \ref{subsec:fp.po}). In Section \ref{subsec:attractor}, we prove Theorem \ref{thm:attractor} on the existence of attracting sets of rhombic tilings mentioned above.  Returning to experimental results,  Section \ref{subsec:rp2} shows numerically how tilings whose parastichy numbers sum up to 4 coexist in the shape space of chains of four primordia. We show that the latter set, for the chosen parameter, has the topology of the projective plane. \\

\section{Numerical Explorations}
\label{section:numexp}
Before formally studying the concepts of rhombic tilings, primordia front, and their relation to the dynamics of the Snow map \snow, we present some numerical observations that motivated our theoretical inquiry.

 \subsection{Asymptotic Behavior of \snow}
\label{subsec:observe}

In our numerical simulations, we consistently observed that, under iterations of \snow, \emph{all}  configurations converge to  of ``fat rhombic tilings": lattice-like sets of points of the cylinder that are vertices of tilings with rhombic tiles that are not too thin (See Section \ref{subsec:tilings}). The tilings have, like the classical  lattices of phyllotaxis, parastichies: strings of tangent primordia winding up and down the cylinder in somewhat irregular helices. And as with lattices, these parastichies come in two families winding in opposite directions (Fig \ref{fig:Periodpen}).

Interestingly, we observed two distinct types of convergence to rhombic tilings: a finite time convergence and an asymptotic (infinite time) convergence. In our experiments, asymptotic convergence always involves at least one pair of pentagonal and triangular tiles, repeating along a parastichy (Fig. \ref{fig:Periodpen}). One can see from the figure (see Proposition \ref{prop:frontperiod} for a proof) that when an orbit goes through segments of a given rhombic tiling of parastichy numbers $M, N$ its shape repeats periodically, with period $MN$. Hence, orbits that we have observed are either periodic (in the shape space), preperiodic or asymptotically periodic (with triangle and pentagon pairs). Moreover we observed large continua of  tiling segments that are periodic orbits. In Theorem \ref{thm:attractor}, we prove the existence of such continua and of the preperiodicity of all orbits near steady state lattices. We will leave the analysis of the asymptotic convergence to a later work.

\begin{figure}[h] 
   \centering
   \includegraphics[height= 3.4in]{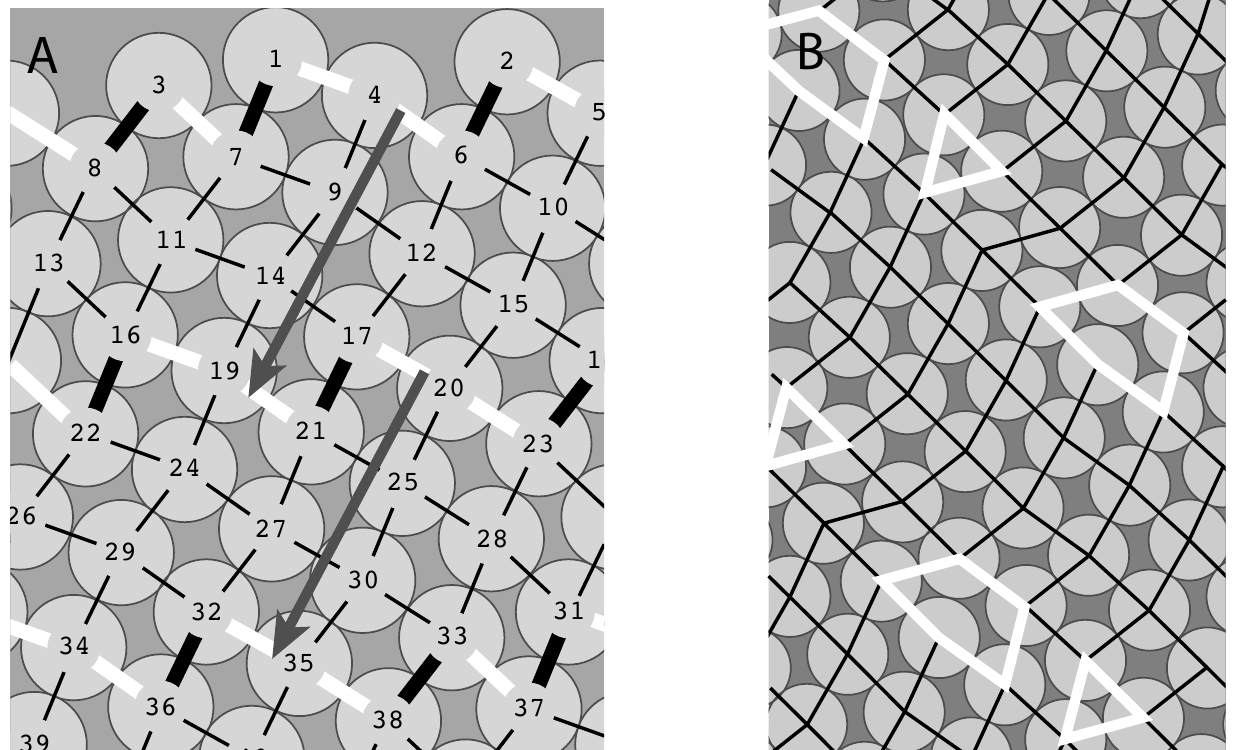} 
   \caption{\footnotesize{ {\sf (A)} Rhombic tiling of parastichy numbers (5,3) obtained by iterating the transformation \snow~on a front of parastichy numbers 5,3. The older the primordium, the greater its index. Three fronts, at primordia 1, 16, and 31, are shown with their up vectors in thick black, down vectors in thick white.  A ``period vector" (shown in dark grey) joins primordia $k$ and $k+15$, and translates a front into another periodically, with $5\times 3$ primordia in between: this is an orbit of period 15 (see Theorem \ref{theorem:RTperiodic} and Proposition  \ref{prop:frontperiod}).  {\sf (B)} Asymptotically periodic orbit, with pairs of triangles and shrinking pentagons aligned in a parastichy.}}
   \label{fig:Periodpen}
\end{figure}

It is intuitively clear that, at a given time step. of the iteration, the top  connected layer of primordia  holds the key to the dynamics and the geometry of the orbits. We call such a layer a primordia front (Section \ref{subsec:parentsEtc}). The number of ``up" and ``down" vectors forming the zigzagging curve as one travels from left to right on a front corresponds  to parastichy numbers in the case of lattices and tilings (Proposition \ref{prop:parastnum}). We call them front parastichy numbers (Section \ref{subsec:parentsEtc}). We contend that counting front parastichy numbers at each step of the iteration - which can be  programmed in either simulations or data analysis - may be less misleading than the divergence angles commonly used in this type of experiment (See Figures \ref{fig:tilingfit} and \ref{fig:divVSparast} and the next section).

\subsection{The Fibonacci Path}
\label{subsec:fibonacci}

 \begin{figure}[h] 
   \centering
   \includegraphics[width=6in]{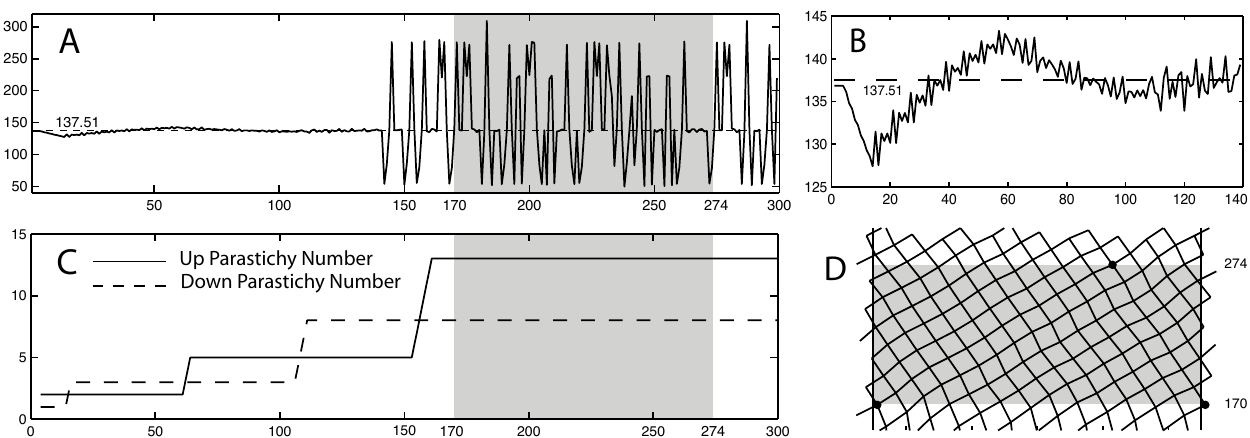} 
   \caption{\footnotesize { {\sf (A)} Divergence angle vs number of iterates of \snow. The initial condition is the $(2,1)$ steady state lattice for the parameter $\diam=.2083$ (diameter of primordia). $D$ is decreased by 1\% at each iterate, until iterate 170, after which it is kept constant. Note how close to the golden angle ($\approx 137.51^0$) the divergence angle is until iterate 140, and how wildly it oscillates after that. A periodicity of $104=8\times13$ can be observed after iterate 170. {\sf (B)} A section from {\sf (A)} blown up to show the oscillation of the divergence angle that mirrors, up to some small amplitude secondary oscillation, the zigzagging of the Fibonacci branch of the bifurcation diagram of Fig. \ref{fig:bifdiag}. {\sf (C)} The front parastichy numbers vs iterate numbers. Note the extreme regularity of this data, contrasting with the oscillations of the divergence angle. {\sf (D)} The $(13,8)$ dynamical tiling obtained after iterate 170, when $\diam$ is kept constant. The mild, nature-like undulations of its parastichies contrast with the irregularities of its divergence angles seen in (A).  }}
   \label{fig:divVSparast}
\end{figure}

The litmus test for a model of phyllotaxis is its ability to reproduce aspects of the bifurcation diagram - or fixed point set - of Section \ref{subsec:fp.po}, and especially of its Fibonacci branch.

This diagram is formed by the generators (see Section \ref{sec:classicgeom}) of the ``good" lattices of phyllotaxis, that are steady states of the given model (in this case presented, of both \snow~and the Hofmeister map of \cite{jns}). Each (dark) curve segment of the diagram corresponds to lattices with a given pair of parastichy numbers whose shapes are fixed under \snow. The $x$ coordinate of a generator corresponds to the so-called divergence angle between two consecutive points in the vertical ordering of the lattice. It also corresponds to the difference of $x$ coordinate of the new primordium and the next in an iteration process. The divergence angle, and its connection to parastichy numbers in lattices, has been widely used to explain and detect in models the Fibonacci phenomenon \cite{douadycouder}.   The Fibonacci branch of the bifurcation diagram is the largest in the diagram, and starts at lattices of parastichy numbers (1,1), corresponding to the beginning of the growth of most monocotyledonous plants. In our work on the Hofmeister map \cite{jns}, we showed that the steady state lattices are attractors, accounting for the fact that once near the Fibonacci branch, a configuration remains near it as the parameter (in that case the internodal distance) was decreased.

 We were originally pessimistic about \snow~yielding Fibonacci transitions as the parameter varies. Indeed, we had observed numerically that a steady state lattice for \snow~is part of an attracting manifold of periodic orbits and the eigenvalues of the differential are either 0 or on the unit circle (A consequence of Theorem \ref{thm:attractor}). Hence the steady states for \snow~can at best be neutrally stable. However, our experiments (Fig. \ref{fig:divVSparast} ({\sf A \& B})) show that, as we lower the diameter of the primordia while iterating \snow,  the Fibonacci phenomenon, as measured by front parastichy numbers, is in fact much more robust in our \snow~model than the divergence angle measurements indicates: while the  divergence angle can vary wildly even in an orbit  close to a lattice,  the parastichy numbers stay constant. Orbits do not have to stay too close to lattices to follow the Fibonacci route: It is sufficient that they stay in a neighborhood the substantially larger and attracting set of rhombic tilings. This flexibility allows for much faster transitions than previously thought, in a time scale observed in plants, as we will show in future work.  Last but not least, the strong attraction of \snow~orbits to the set of rhombic tilings should make this set persist topologically in nearby systems.

\section{Classical Geometry of Phyllotaxis}
\label{sec:classicgeom}
\subsection{Underlying Geometry} In this paper, we concentrate on cylindrical phyllotaxis. We normalize the cylinder $\cyl$ to have circumference 1.  Mathematically, $\cyl $ is the cartesian product $\circle\times \real$ of the unit circle $\circle= \real/\intgr$ with the reals. Note that fixing the circumference of the cylinder does not mean that we preclude lateral  plant growth in our modeling. We make this convenient normalization choice without loss of generality since, in the patterns we study, the important parameters (such as the ratio $D$ of the size of primordia relative to the diameter of the meristem) are independent of scale. Both botanists and mathematicians often unroll cylindrical patterns on the plane $\real^2$, which can also be seen as the complex plane $\cmplx$. This is the covering space of the cylinder (see Section \ref{subsec:cover}).  We will use the same notation for points and vectors in $\real^2$ and $\cyl$. By a \emph{configuration}, we mean a finite set of points in $\cyl$ ordered by height. These points represent centers of primordia along the stem.

\subsection{Covering Space Notions and Notation}
\label{subsec:cover}
We often describe objects in the cylinder via their covers and lifts in the plane. The intuitive notion of cover of a set, in the case of the cylinder is simple: mark each point of the set with ink, and use the cylinder as a rolling press. As you roll the cylinder indefinitely on the plane, the points printed form the cover of the original set. Each piece of the cylindrical pattern is repeated at integer intervals along the $x$-direction. The cover of a helix, for example, is a collection of parallel lines. The lift of a helix at a point is the choice of one of these lines.

Here is a more rigorous description of these classical concepts \cite{munkres} and notation that we will be using.   The natural projection $\pi: \r2  \mapsto\cyl$ which maps a point $(x,y)$ to $(x \text{ mod } 1, y)$ is a \emph{covering map} and the plane $\r2$ is a \emph{covering space} of the cylinder $\cyl$. This means that $\pi$ is surjective, and that around any point $z$ of $\cyl$, there exists an open  neighborhood $U$ such that $\pi^{-1}(U)$ (the inverse image of $U$) is a disjoint union $\cup U_k$ of open sets of the plane each homeomorphic to $U$. One says that $\pi$ is a \emph{local homeomorphism} and 
 that $U$ is \emph{evenly covered}. In the case of the cylinder, $\pi$ is also a local isometry, for the metric induced by $\pi$ on the cylinder. A subset $X$ of $\r2$ is a \emph{fundamental domain} if $\pi: X\mapsto \cyl$ is a bijection. Any region of $\r2$ of the form $\{(x, y)\in \r2 \mid a\leq x<a+1\}$  is a fundamental domain.
 
 The \emph{cover} of a subset $Y$ of $\cyl$ is the inverse image $\tilde Y = \pi^{-1}(Y)$ of $Y$.   A set $\tilde Y$ of the plane is a cover of its projection $\pi(\tilde Y)$ if and only  if $\tilde Y+ (1, 0) = \tilde Y.$
The ``tilde" notation as above is often used to denote covering spaces. In this paper, we also use the underline notation to denote the projection of a set in the plane to the cylinder: $\underline X = \pi(X)$.

As with all covering maps, $\pi$ has the lifting property: if  $\gamma$ is a path in $\cyl$ and $c\in \r2$ is a point ``lying over" $\gamma (0)$ (i.e. $\pi(c) = \gamma(0)$), then there exists a unique path $\rho \in \r2$  lying over $\gamma$ (i.e. $\pi \circ \rho = \gamma$) and with $\rho(0) = c$. The curve $\rho$ is called the \emph{lift of $\gamma$ at $c$}.The lift of a path is only a connected part of its cover: for instance the lift at $(3, 0)$ of the line of equation $\underline x = 0$ of the cylinder is the line $x=3$ in the plane, whereas its cover is the union of all the lines $x=k, k\in \intgr.$

\subsection{Cylindrical Lattices, Helical Lattices, Multijugate Configurations }
\label{subsec:lattices}
 A \emph{cylindrical lattice} $L$  is a set of points in $\cyl$ whose cover $\tilde L$ is a lattice of $\r2$:
 
  $$ \tilde L = \left\{ m \vec v+n \vec w \in \real^2 \mid  m,n\in \intgr\right\},$$ 
where $\vec v,\vec w \in \r2$ are independent generating vectors.   Note that $\tilde L$ a discrete subgroup of $\r2$ isomorphic to $\intgr^2$.
Since $\tilde L$ is a cover, it must be invariant under translation by $\2vec 10$. Changing bases if necessary, one can assume that $\vec w = \left(\frac 1k, 0\right)$ for some positive integer $k$, called the \emph{jugacy} of the lattice.  

If  $k = 1$, we say that $L$ is \emph{monojugate} or that it is a \emph{helical lattice}. In this case $\vec w = (1,0) = (0,0) \mod 1$ and  $L$ has the unique generator $\vec v$. If $k>1$, $L$ is called a \emph{multijugate configuration} or specifically a $k-$\emph{jugate} configuration (or $k$-jugate lattice). A cylindrical lattice $L$ is a discrete subgroup of $\cyl = \circle\times \real$ isomorphic to $ \intgr\times\intgr/{k\intgr}$ (simply $\intgr$ in the case of a helical lattice). In a $k$-jugate lattice, each point is part of a set of $k$ points, called a \emph{whorl}, evenly spread around a horizontal circumference of $\cyl$.

 Parastichies of a helical lattice $L$ are helixes joining each point of $L$ to its nearest neighbors. 
We now make this more precise. In general, there are two points of $\tilde L$ nearest to $0$ in the positive half plane. Say $z_M= M\vec v + (\Delta_M, 0)$  and $z_N= N\vec v + (\Delta_N, 0)$  nearest to 0, where $M,N, \Delta_M, \Delta_N\in \intgr$. Also assume that $\widevec{0z_M}$ makes a larger angle with the horizontal than $\widevec{0z_N}$, so in particular $\widevec{0z_M}$ and $\widevec{0z_N}$ are not colinear. The line through 0 and $z_M$ lifts a helix in $\cyl$  that contains all the points $\pi(kz_M)=\pi(kM\vec v), k\in \intgr$. The set of these points is called a \emph{parastichy}, and the helix connecting them a \emph{connected parastichy}. There are $M$ helixes, also called parastichies, parallel to this one. Each goes through a set of points $\left\{\pi(kz_M+jz_N)\right\}_{k\in \mathbb Z}$, for a fixed $j\in \{0, \ldots, M-1\}$. To prove this, one shows that when $j=M$, one obtains another lift of the original parastichy, using 
\begin{equation}
\label{eqn:MNcoprime}
N\Delta_M-M\Delta_N= 1
\end{equation}
This is a consequence of $z_M, z_N$ being closest to 0 (\cite{jns}, Proposition 4.2) and it implies that $M$ and $N$ are co-prime ($\gcd(M,N) = 1$).  Thus there are $M$ parastichies and they  correspond to a cosets of the subgroup $M\intgr$ of $\intgr$. Likewise, there are $N$ parallel parastichies joining the second closest neighbors. 

\begin{figure}[h] 
   \centering
   \includegraphics[width=6in]{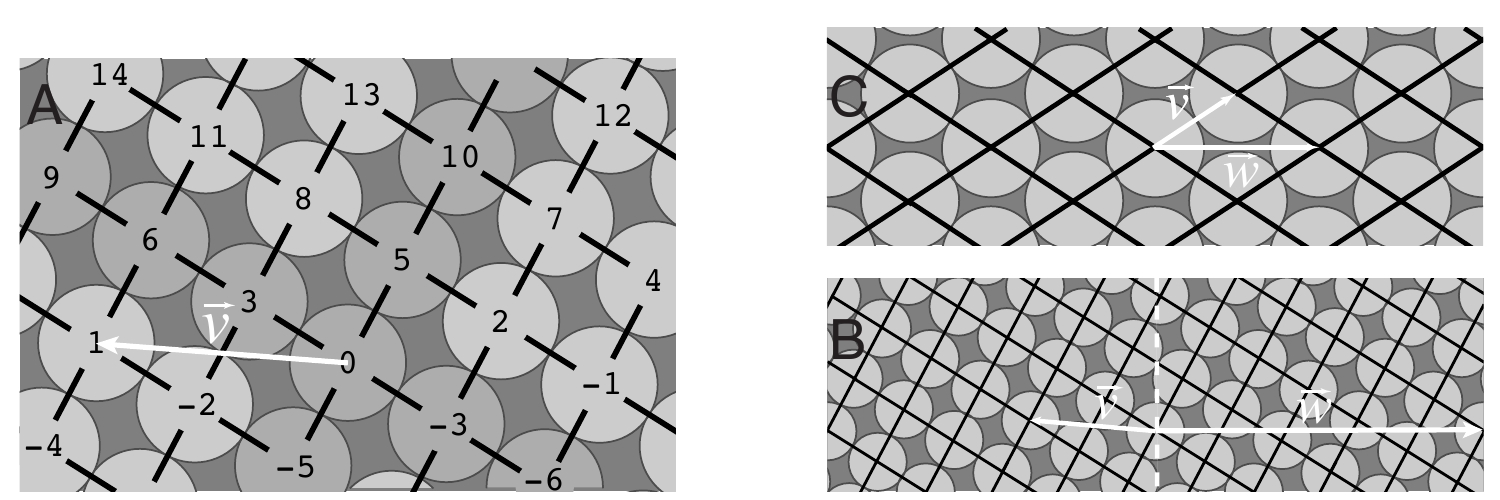} 
   \caption{\footnotesize {\bf Cylindrical Lattices.} Each represented in a fundamental domain of their cover with $-\frac12<x\leq \frac 12$. All these lattices are rhombic: each primordium is equidistant to its 4 nearest neighbors. {\sf (A)}  3, 5 helical lattice, with only one generator $\vec v$. We only show the indices $k$ of the points $z_k$. Note how $z_3$ and $z_5$ are the closest to $z_0$. We have shaded the parastichies through $z_3$ and $z_5$.  There are 5 parastichies parallel to that through $z_5$, and 3 parallel to that through $z_3$.    {\sf (B)}  2-jugate (bijugate) 6,10 lattice, obtained by rescaling two copies of the one in (A) by  $1/2$, and setting them side by side on the cylinder. This lattice has two generators, $\vec v$ which is half the vector $\vec v$ of $A$, and the vector $\vec w = (1/2,0)$.  each primordium is in a whorl of 2, separated by $\vec w$.  {\sf (C)}  4,4 lattice (4-jugate), with two generators  $\vec v$ as shown  and $\vec w = (1/4, 0)$. Each primordium is in a whorl of 4. }    \label{fig:Lattice&Whorls}
\end{figure}

If $L$ is a $k$-jugate lattice, we can trace parastichies through nearest neighbors in a similar fashion. This time the parastichy numbers $M$ and $N$ must have the common divisor $k$. An intuitive way to see this is that, rescaling the cover $\tilde L$ of  $L$ by $k$, one obtains the cover of a helical lattice, call it $L_h$. Build the parastichies through nearest neighbors for $L_h$ as before. Then rescale back by $1/k$ - the rescaled parastichies are parastichies of $L$. Since you need $k$ copies of the rescaled $L_h$ to go around the cylinder, you need $k$ rescaled copies of each parastichy of $L_h$, crossing the $x$-axis at intervals of $1/k$,  to get all the parastichies of $L$.
 
Thus, all cylindrical lattices can be classified by their parastichy numbers $(M,N) = k(i,j)$ where $k$ is the number of primordia in a whorl. Helical lattices are the special case where $k=1$. \emph{Whorled} configurations,  where primordia in a new whorl are placed midway between those of the previous whorl, is another notable case, which corresponds to $k(1,1)$.  

\subsection{Limitations of Cylindrical Lattices in Phyllotaxis.}
In Section \ref{subsec:fibonacci}, we presented a numerical simulation showing Fibonacci transitions along orbits of \snow~when the parameter $D$ is decreased. We argued briefly that the dynamical transitions observed  mirrored the continuous geometric deformation of helical lattices along the main Fibonacci branch of the bifurcation diagram. The existence of the (connected) Fibonacci branch has been the basis of many explanations of the phenomenon since the $19^{th}$ century\footnote{This neat \emph{geometric} fact has often been a source of confusion between the global deformation of a pattern and its  transitions via a \emph{dynamical} process with varying parameter.} (\cite{weisse}, \cite{vaniterson}).   Unfortunately, this kind of argument, made rigorous for the Hofmeister model in \cite{jns}, cannot work for transitions involving a change of jugacy in the pattern. One of these transitions, from $2(1,1)$ (decussate) to $(2,3)$ (Fibonacci spiral) is the norm in the vast majority of dycotyledonous plants such as the sunflower, where after a few whorls of two leaves at $90^o$ angle, symmetry is broken and a spiral pattern emerges. This transition cannot be attributed to the proximity of iterated patterns to a continuous path \emph{within} the set of lattices between lattices of parastichy numbers $(2,2)$ and $(2,3)$. Indeed, no such continuous path exists, since it would have to involve the continuous deformation of the vector $w = (1/2, 0)$ into $(1,0)$ within the discrete set of rational vectors of the form $(1/k, 0)$, which is clearly absurd.

Even in a Fibonacci transition, the global geometric deformation of lattices (orthostichies becoming parastichies when $D$ decreases) does not translate easily into a dynamical understanding of the transitional region.  In short, we need more flexible and local geometrical tools to better describe  dynamical transitions.

\section{New Geometry for Phyllotaxis}
\label{sec:newgeom}
 In this section, we introduce  primordia fronts and \ph~tilings.  They address the limitations noted in the previous paragraph.   Fronts are  local in nature and are well defined in the setting of general configurations of points of the cylinder. We will show that fronts are key in understanding transitions. \PH~tilings and more specifically rhombic  tilings allow many more deformations than cylindrical lattices while still featuring parastichies.  We give an algebraic definition of these tilings as a set of points obtained by cyclically adding ``up" and ``down" vectors. In later sections, we derive the geometric and periodicity properties of these tilings and of the tiles they bound.

\subsection{Parents, Ontogenetic Graphs, Fronts, Local Parastichy Numbers}
\label{subsec:parentsEtc}
This subsection gives definitions regarding very general configurations of points of the cylinder. They can naturally be adapted to other geometries (cone or disk) as well (see \cite{jpgr}). We consider general configurations of a  number $K$ of disks of a given diameter $D$ in the cylinder.  These  configurations are given by their centers $(p_1, \ldots , p_K)$ and they form the set $\cyl^K$,  Cartesian product of $K$ copies of the cylinder. Occasionally, we need to consider countably infinite configurations as well.

 The \emph{ontogenetic order} for a configuration in $\cyl^K$ is a choice of indices $\{1, \ldots, K\}$  for the points which corresponds to the following order of the points coordinates:
$$  i > j \sameas y(p_i) < y(p_j)  \text{ or }\{ y(p_i) = y(p_j) \text{ and } x(p_i) > x(p_j)\}, $$ where we choose the fundamental domain $x \in \left(-\frac 12, \frac 12 \right] $. 

Often, such  as with lattices, we consider finite configurations that are pieces of infinite ones. A configuration ${\config}\in \cyl^K$  comprising all the points of an infinite configuration $X$ between some $p_i$ and $p_{i+K}$ in the ontogenetic order of $X$ is called a \emph{segment} of $X$ of length $K$.

 A primordium $p_j$ is a \emph{left (resp. right) parent} of $p_i$ if it is tangent below and to the left (resp. right) of $p_i$.  More precisely, we adopt the convention that, for $p_j$ to be left parent of $p_i$, the coordinates $x, y$ of the vector $\widevec{p_ip_j}$ must satisfy $-1<x<0, y\leq 0$ and $x^2+y^2 =D^2$, and $1>x\geq 0, y<0, x^2+y^2 =D^2$  for $p_j$ to be right parent of $p_i$. 
 In the obvious fashion,  $p_i$ is a right (resp. left) child of $p_j$ if $p_j$ is a left (resp. right) parent of $p_i$.
 
 The \emph{ontogenetic graph} of a primordia configuration is the directed graph embedded in $\cyl$ whose vertices are the centers of the primordia and where oriented edges are drawn between primordia and their parents (if they have any).

Given an ontogenetically ordered configuration $\config$  of $\cyl^K$, we call \emph{parents data} the information about which primordia are parents of which primordia. One way to represent this data is by a $K\times K$  \emph{parents data matrix}, whose $(i,j)^{th}$ entry   is $1$ if $p_j$ is left  parent of $p_i$, $-1$ if $p_j$ is right  parent of $p_i$ and 0 if $p_i$ is not a parent of $p_j$.  Note that the absolute value of this matrix is just the adjacency matrix of the (directed) ontogenetic graph.

A \emph{primordia chain} for a configuration is a subset $\{p_{i_1}, \ldots p_{i_q}\}$ of distinct points in the configuration such that: 

\begin{itemize}
\item{} The chain is connected by tangencies: for all $k \in {1, \ldots, q}$, primordium $p_{i_{k+1}}$ is either a right parent or right child of $p_{i_k}$.
A chain can thus be represented by a piecewise linear curve through the centers of its primordia, which can be lifted to $\r2$.
\item{} The chain is  \emph{closed} and does not fold over itself: the point $p_{i_q}$ is either a left parent or left child of  $p_{i_1}$ and any lift at a point $P$ with $\pi(P)=p_{i_1}$ of the chain is the graph of a piecewise linear function over the $x$ axis in $\r2$ joining $P$  to its translate $P+ (1,0)$.

\end{itemize}

The vector  $\widevec{p_{i_{k}}p_{i_{k+1}}}$ is an \emph{up vector} of the chain if $p_{i_{k+1}}$ is  a right child of $p_{i_{k}}$.
The vector   $\widevec{p_{i_{k}}p_{i_{k+1}}}$ is a \emph{down vector} if $p_{i_{k+1}}$ is  a right parent of $p_{i_{k}}$.
We call the number of up (\resp down) vectors in a chain its \emph{right} (\resp \emph{left}) parastichy number.
If $p_{i_{k+1}}$ is always parent of   $p_{i_{k}}$ for $k=1, \ldots m-1$ and then always a child for $i = m , \ldots, q$,  we call the chain a \emph{necklace}. \\

Given a configuration ordered ontogenetically, a \emph{front at $k$} is a chain with primordia of indices greater or equal to $k$, such that any  primordium (not necessarily in the configuration) which is the child of a primordium in the chain, without overlapping any other primordium in the chain, is necessarily at a height greater or equal to that of $p_k$. The parastichy numbers of a front are called \emph{front parastichy numbers}. 


 \remark{ Most of the notions defined above are applicable to plant data by relaxing the definition of left (resp. right) parent to that of ``closest primordia below to the left (resp. to the right)" with some tolerance level. In the case of configurations on the disk, ``below" translates to ``farther away from the meristem" (see \cite{jpgr}). Algorithms using these notions were also used to produce Fig. \ref{fig:tilingfit} and \cite{fig:divVSparast}.}

\subsection {Phyllotactic Tilings}

\label{subsec:tilings}

A \emph{\PH~tiling}  is a set of points of $\cyl$ that can be obtained by summing to a base point cyclically ordered sums of ``up'' and ``down'' vectors. More precisely, a tiling $\tiling$ is determined by a base point $(a,b)\in \r2$, \emph{down vectors} $\dd_1, \ldots, \dd_M \in \r2$  where each $\dd_k$ has components $x\geq 0,  y< 0$,  and \emph{up vectors} $\uu_1, \ldots, \uu_N \in \r2$, each with components $x>0, y\geq 0$.  Moreover, we ask that
\begin{equation}
\label{eqn:U+D=1}
\sum_{j=1}^M \dd_j + \sum_{i=1}^N\uu_i = (1,0)
\end{equation}

We then define 
$$\tiling= \{\zz_{m,n}\in \cyl \mid m,n \in \intgr\},$$
 where   $\zz_{m,n} = \pi(z_{m,n})$ with:
 
 $$
z_{m,n}= (a,b) +  sumdown(m)+sumup(n) 
$$
and where
\begin{equation}
\label{eqn:sumupdown}
sumdown(m) =\left\{\begin{array}{ccl} \sum_{j=1}^m \dd_j &\text{if}& m>0\\ 
0 &\text{if}&m = 0\\
 \sum_{j=0}^{|m|-1} -\dd_{M-j} &\text{if}& m<0
 \end{array}\right. 
\end{equation}
and we use the periodicity convention $\dd_{j+M}=\dd_j,  \forall j \in \intgr$. The function $sumup(n)$ is a cyclical sum of up vectors defined similarly, with the convention that $\uu_{i+N}=\uu_i,  \forall i \in \intgr$.   The numbers $M,N$ of down and up vectors are called the \emph{parastichy numbers of the tiling}, a terminology justified by Proposition \ref{prop:paras}.

The \ph~tiling  $\tiling$ is \emph{rhombic} if all the up and down vectors have same length $\diam$, and we call $\diam$ the \emph{length} of the tiling.
A \emph{fat  tiling} is a \ph~tiling such that the angles between \emph{any} two down and up vectors are in the interval $[\pi/3, 2\pi/3]$.  It is not hard to see that the ontogenetic graph of a rhombic tiling is the embedded graph in the cylinder whose vertices are the points of the tiling and the edges are the down and negative up vectors connecting them. For general  \ph~tiling, we call this graph the \emph{graph of the tiling}. We call the connected components of the complement of the graph its \emph{tiles}.

\remark{The following are simple but important consequences of the previous definitions:
\begin{itemize}
\item Equation \eqref{eqn:U+D=1} implies that the set $\tilde\tiling= \{z_{m,n}\in \r2 \mid m,n \in \intgr\}$  is indeed the cover of the tiling $\tiling$: if $z \in \tilde\tiling$, so does  $z + (1,0).$
\item Equation \eqref{eqn:U+D=1} also  implies that $\zz_{m+kM,n+kN} = \zz_{m,n}$ for all $n, m, k \in \intgr$.  It also implies that the down vectors stringed together, followed by the up vectors form a necklace, which clearly has parastichy numbers $N, M.$
\item If all the down vectors $\dd_j= \dd$ are equal are equal and if  up vectors $\uu_i= \uu$ are equal, and if the tiling parastichy numbers $M, N$ are coprime the tiling is in fact a lattice  of parastichy numbers $M,N$, generated by $\dd$ and $\uu$. If the tiling is fat, $\dd = -z_M$ and $\uu = z_N$ in the notation of Section \ref{subsec:lattices}.  If $M, N$ are not coprime, the tiling is a multijugate configuration.
\item The condition of ``fatness" implies that a rhombic tiling can be seen as a configuration of tangent disks with no overlap. 
\end{itemize}
\label{remark:tiling}

 \begin{figure}[h] 
   \centering
   \includegraphics[width=6in]{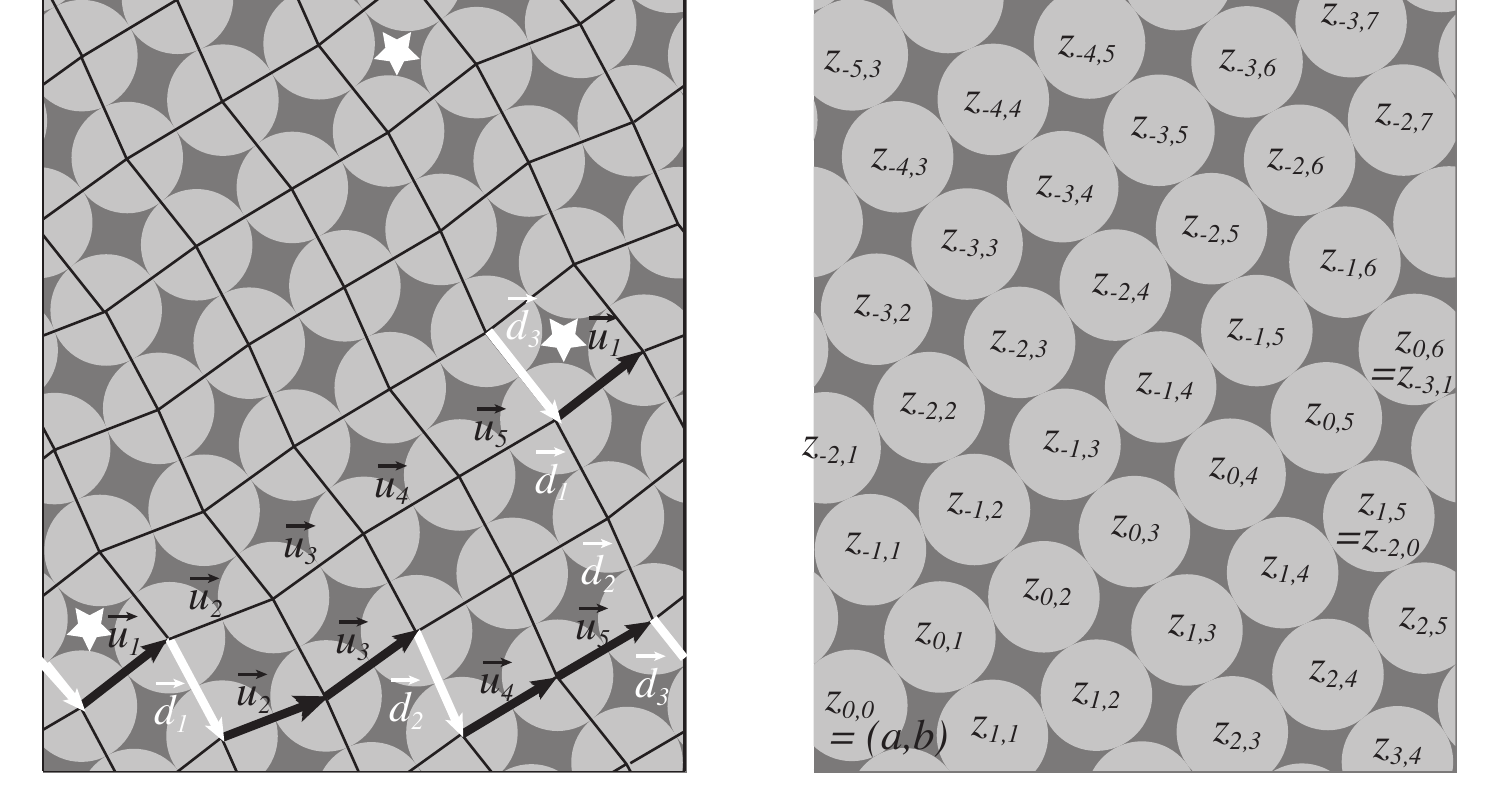} 
   \caption{\footnotesize { A fat rhombic tiling of parastichy numbers $M=3, N=5.$ On the left, the up and down vectors are shown, on a front, and on a necklace.  Translating the tiling from one star to another does not change it: If a pair of down and up vectors meet at a point, they meet again at the other extremity of the necklace they both belong to, inducing periodicity in the tiling (Theorem \ref{theorem:RTperiodic}).  On the right, the $\intgr^2$-like numbering of primordia of the same tiling. The equation  $\zz_{m+kM,n+kN} = \zz_{m,n}$ is shown at primordia $z_{1,5}$ and $z_{0,6}$. Note that parastichies are obtained by keeping all the points that have same down (\resp up) index. }}
   \label{fig:orm}
\end{figure}
Given a certain type of subsets of the cylinder, we say that two subsets have the \emph{same shape} if one is a translation of the other.  We call the set of all distinct classes of shapes the \emph{shape space}. 

\proposition{When $(M+N)\diam>1$, the set $\rtmnd$ of opposed rhombic tilings of parastichy numbers $(M,N)$ with vectors of length $\diam$ is  a manifold of dimension $M+N$, possibly with boundaries and corners. Every strictly fat lattice is contained in an open neighborhood of such a manifold. The set of fat   rhombic tilings is a submanifold (with possibly more boundaries and corners) of $\rtmnd$. The set $\rtmndk$ of segments of length $k>M+N$ of tilings in $\rtmnd$ is homeomorphic to $\rtmnd$. When considering the shape spaces of these respective types of objects, subtract 2 to the dimension of each set above.}

\proof We parameterize $\rtmnd$ by the two independent variables for the base point $(a,b)$, and $M+N-2$  angles with the horizontal of  $M-1$ down and $N-1$ up vectors. The last down and up vectors are given by the two sides of an isosceles triangle of equal sides of length $\diam$ between the points $(1,0)$  and $\sum_{j=1}^{M-1} \dd_j + \sum_{i=1}^{N-1}\uu_i$.  Boundaries are determined by the inequalities  $-\frac{\pi}2\leq \theta_j<0$ on the angles of down vectors, $j\in \{1, \ldots, N-1 \}$ and $0\leq \upsilon_j <\frac\pi 2, \quad  j\in \{1, \ldots, M-1\}$ for the angles of up vectors, as well as some more complicated inequalities (that we won't make explicit here) involving differentiable functions of these $M+N-2$ angles that guarantee that the last up and last down vectors can be defined and satisfy the above inequalities as well. The statement on lattices derives from the fact that, for  a strictly fat lattice, the inequalities on angles stated above are all strict. The condition of fatness only adds further inequalities on the angles $\theta_j, \upsilon_i$. The set of segments of tilings of sufficient length is parameterized by the same angles, with the same inequalities as the set of  itself, and thus these sets are homeomorphic. When considering the shape spaces, the base point is removed from the parameterization, lowering the dimension by 2. \qed
\label{prop:dimension}

\subsection{Periodicity of \PH~Tilings }

We say that an infinite configuration of points $X$ in $\cyl$ or $\r2$ has \emph{period vector} $\vec V$ if $X+\vec V = X.$

\theorem{\bf (Periodicity of Tilings) }{The cover $\tilingcov$ of a \ph~tiling $\tiling$ with down and up vectors and   $ \dd_1, \ldots, \dd_M$ and $ \uu_1, \ldots, \uu_N$ has the two independent period vectors  $\vec D = \sum_{j=1}^M \dd_j$ and $\vec U = \sum_{i=1}^N \uu_i$. Since in $\cyl$ these two vectors sum to 0,  $\tiling$ has only one independent period vector, say $\vec U$.}
\label{theorem:RTperiodic}
\begin{proof} Pick a point $z_{m,n}\in \tilingcov$. We will do the case $m,n >0$, the other cases derive from this one and the equality $z_{m+kM,n+kN} = z_{m,n}+(k,0)$. From $\uu_{j+N} = \uu_j$, we obtain:
\begin{eqnarray*}
z_{m,n} + \vec U &=& (a,b) +\sum_{i=1}^m \dd_i +  \sum_{j=1}^n  \uu_j + \sum_{j=1}^N  \uu_j\\
&=& (a,b) +\sum_{i=1}^m \dd_i +  \sum_{j=1}^{n+N}  \uu_j = z_{m,n+N}  \in \tilingcov
\end{eqnarray*}
Likewise, we obtain $z_{m,n} + \vec D = z_{m+M,n} \in \tilingcov.$ This proves that $\tilingcov + \vec U \subset \tilingcov$ and $\tilingcov + \vec D \subset \tilingcov$. Inclusions in the other direction are proven identically, by subtracting the vectors $\vec U, \vec D$ to points of $\tilingcov$ and showing that one obtains points of $\tilingcov$. The statement about the periodicity of $\tiling$ already contains its proof. \end{proof}

The periodicity above can be expressed by saying that \ph~tilings are multilattices.
 A \emph{multilattice} $\Lambda$ of $\cyl$ or $\r2$ is the union of a number $k$ of copies of the same lattice $L$, each translated by a different vector $v_i$:
$$
\Lambda = \bigcup_{i=1}^k (\vec v_i+L).
$$
} 

Note that the generating vector(s) of the lattice $L$ is (are) period vector(s) of the multilattice.  To see that a \ph~tiling is indeed a multilattice, let $L$ be the cylindrical lattice generated by  $\vec U = \sum_{k = 1}^{N} \uu_k$ and let  the $v_i's$ be vectors of the form $\sum_{i=1}^m \dd_i +  \sum_{j=1}^{n},\uu_j,  m\in \{1, \ldots,M\}, n\in \{1, \ldots,N\}$. This interpretation of tilings as multilattices explains the appearance of ``parallelogram" shaped tiles of ``size" $M, N$ that repeat periodically in a tiling.

\subsection{Parastichies of \PH~Tilings} 

We call \emph{parastichies} at a point $\zz_{m,n}$ of a tiling in $\cyl$ the two subsets of the tiling obtained by successively adding or subtracting all the up vectors (starting with $\uu_{m+1}$ and following the cyclical order) or all the down vectors (starting at $\dd_{n+1}$ and following the cyclical order). We call these parastichies \emph{left and right parastichies} respectively, denoting their directions as they are traversed down from the point $\zz_{m,n}$.
A \emph{connected parastichy} is the piecewise linear  curve formed by joining the successive parastichy points with the up or down vectors that connect them. We denote by $LP_{m,n}$ and $RP_{m,n}$ the lifts of the left and right connected parastichies through  $z_{m,n}.$ Hence $LP_{m,n}$ is the piecewise linear curve through the points $\{z_{m,j} \mid j\in \intgr\}$ and $RP_{m,n}$ is the piecewise linear curve through the points $\{z_{i,n} \mid i\in \intgr\}$, where two successive points are joined by an up (resp. down) vector.
\\
In the case that the tiling is a fat lattice or multijugate configuration, the above definition coincides with the classical definition of parastichy: the regular helices joining nearest left  (resp. nearest right) neighbors. 
\\
In the case that the tiling is a fat lattice or multijugate configuration, the above definition coincides with the classical definition of parastichy: the regular helices joining nearest left  (resp. nearest right) neighbors. 

The following underlines the similarity between tilings and lattices and justifies the qualifier of rhombic given to some of the tilings we consider:

\proposition{{\bf(Properties of lifted parastichies)} \label{prop:propliftparas} The lift of a connected left parastichy $LP_{m,j}$ intersects the lift of a connected right parastichy $RP_{i,n}$ at the unique point $z_{m,n}$.  The curves $LP_{m,n}$ and $LP_{i,k}$ intersect if and only if $m=i$ (in which case they are equal); $RP_{m,n}$ and $RP_{i,k}$ intersect if and only if $n=k$ (in which case they are equal). 

The tiles of the cover of  a   \ph~tiling are parallelograms, whose vertices are of the form $z_{m,n}, z_{m-1,n}, z_{m,n-1}, z_{m-1,n-1}$ for some integer pair $m, n$. In the case of a rhombic tiling, these tiles are rhombi.}

\proof The curves  $LP_{m,j}$ and $RP_{i,n}$ contain $z_{m,n}$.  We show that they do not intersect in any other point. If they crossed at another point of the tiling cover, there would exist integers $k,l$ such that $z_{m,k}= z_{l,n}$. But this implies 
$\sum_{j=m+1}^l \dd_j = -\sum_{i=k+1}^n\uu_i$ (we've assumed $l>m>0, n>k>0$, other cases are similar) which is absurd in the plane since up and down vectors are in different quadrants. The case where the parastichies cross at segments between tiling points would yield an equally absurd equality between a linear combination of up vectors and a combination of down vectors - with real coefficients this time. 

\lemma{In the lift of a   \ph~tiling, each curve $LP_{m,n}$  is homeomorphic to a line and  separates $\r2$ into two unbounded regions homeomorphic to a half plane, containing the respective subsets  $\{ z_{i,j},  i<m, j\in \intgr\}$ and  $\{ z_{i,j}, i>m, j\in \intgr\}$ of the tiling. And similarly for $RP_{m,n}$.}

\proof (of the lemma).  
$LP_{m,n}$ is a periodic perturbation of the line $L_l: t \mapsto z_{m,n}+ t U$ where $U=\sum_{i=1}^N \uu_i$, whereas  $RP_{m,n}$ is a periodic perturbation of the line  $L_r: t \mapsto z_{m,n}+ t[(1, 0)- U]$, and thus lifts of connected parastichies have the same asymptotic directions as the corresponding lines $L_l, L_r$. Since all up vectors are in the same quadrant, the orthogonal projection on $L_r$  of $LP_{m,n}$  is a homeomorphism which can be extended to an isotopy (bijective, continous deformation) of the plane, and similarly for the left parastichy. We sketch the isotopy, leaving the details to the reader: draw lines perpendicular to $L_r$  through the parastichy points $z_{m,j}, j\in \intgr$. The lines separate the plane into parallel strips. Apply a shear within each strip so that the vector $\widevec {z_{m,j}z_{m,{j+1}}}$ becomes parallel to $L_r$, translating the other strips so that the transformation is continuous. 
The points  $\{ z_{i,k} \mid  i<m, k\in \intgr\}$ are all on one side of $LP_{m,n}$ through $z_{m,n}$: the right connected parastichy of each point $z_{i,k}, i<m,$ crosses   $LP_{m,n}$ at $z_{m,k}$, and thus at no other point. In particular, the point  $z_{i,k}, i<m$ is on the same side of  $LP_{m,n}$ as $z_{m-1,k}.$ Since the determinants of the angles between the vectors $\dd_{m}= \widevec{z_{m-1,k}z_{m,k}}$ and the vector $\uu_{k-1}= \widevec{z_{m,k}z_{m,k-1}}$ of  $LP_{m,n}$ are of the same sign for all $k$, the points  $z_{m-1,k}$ are all on the same side of  $LP_{m,n}$. Thus $\{ z_{i,k}, i<m, k\in \intgr\}$  are all on one side of $LP_{m,n}$.   Similar statements hold for right parastichies.
\qed
 \label{lemma:parasthalfplane}

Back to the main proof, from the lemma, we obtain that $LP_{m,n}$ and $LP_{i,k}$  intersect if and only if $i=m$, in which case they coincide. Moreover, one left connected parastichy is a translate of another, as it is easy to check that (assuming $m>i$) $LP_{m,n} = LP_{i,k} + \sum_{j=i+1}^m\uu_i$. Similarly for right parastichies. 

Consider now the tiles adjacent to a point $z_{m,n}$. Since there are only four edges of the graph of the tiling adjacent to $z_{m,n}$, there are at most four tiles adjacent to that point. Without loss of generality, we only consider the tile sharing the edges $\widevec{z_{m-1,n} z_{m,n}}=\dd_m$ and $\widevec{z_{m,n-1} z_{m,n}}=\uu_n.$ We claim that this tile has exactly the vertices $z_{m,n}, z_{m-1,n},$ $z_{m,n-1},$ and $ z_{m-1,n-1}$. These points form the vertices of a parallelogram or a rhombus if the tiling is rhombic, with edges $\dd_m$ and $\uu_n$. To check that this indeed forms a tile, there remains to show that no other than these four tiling points is included in this parallelogram. The parallelogram is in a quadrant formed by the parastichies $LP_{mn}$ and $RP_{mn}$ and containing the point $z_{m-1,n-1}$. From Lemma \ref{lemma:parasthalfplane}, any other point of the tiling contained in the parallelogram must be of the form  $z_{i,j}$ with  $i\leq m, j\leq n$. But the parallelogram is also contained in the quadrant formed by the parastichies $LP_{m-1,n-1}$ and $RP_{m-1,n-1}$ and containing the point $z_{m,n}$, which forces $z_{i,j}$ to satisfy $i\geq m-1, j\geq n-1$. In other words, the only points of the tiling that the parallelogram may contain are its already defined vertices. \qed

\proposition{{\bf (Properties of parastichies)} Connected left parastichies $\underline{LP}_{m,n}$ and $\underline{LP}_{i,k}$ cross only if and only if $i =  m \mod M$, in which case they are equal. Thus there are $M$ left parastichies. Likewise $\underline{RP}_{m,n}$ and $\underline{RP}_{i,k}$ cross if and only if $k =  n \mod N$ in which case they are equal,  and there are $N$ right parastichies. Left and right parastichies $\underline{LP}_{m,n}$ and $\underline{RP}_{i,k}$ cross at the points $\zz_{m+qM, k+pN},  q, p \in \intgr$. Interspaces of a \ph~tiling are parallelograms.}
\label{prop:paras}
\proof Since their lifts do not intersect non trivially by Proposition \ref{prop:propliftparas}, left parastichies can only coincide or be disjoint.    $\underline{LP}_{m,n}$ has lifts $LP_{m,n}+q(1,0) = LP_{m+qM,n+qN}, q \in \intgr$ by Remark \ref{remark:tiling}.  Likewise Parastichy $\underline{LP}_{i,k}$ has lifts $LP_{i+pM,k+pN}, p \in \intgr$. The two parastichies coincide if and only if $m+qM = i+pM$ for some $p, q \in \intgr$, \ie if $i =  m \mod M$. There are thus $M$ distinct parastichies. The proof is identical for right parastichies.

Parastichies $\underline{LP}_{m,n}$ and $\underline{RP}_{i,k}$ intersect at projections by $\pi$ of intersection points of two of their lifts  $LP_{m+qM,n+qN}, q \in \intgr$  and $RP_{i+pM,k+pN}, p \in \intgr$. By Proposition \ref{prop:propliftparas}, these intersection points are $z_{m+qM, k+pN},  q, p \in \intgr.$ 

 A tile of the cover of a \ph~tiling is entirely in a fundamental domain  and is thus isomorphic to its projection on the cylinder, which must thus be a parallelogram and  a tile. All tiles arise this way. \qed\\

\subsection{Front vs. Tiling Parastichy Numbers}
\label{subsec:parastnum}
The following proposition connects the three notions of parastichy numbers encountered so far. The following proposition, about chain parastichy numbers, clearly applies to the special case of front parastichy numbers.

\proposition{In a   \ph~tiling with $M$ down and $N$ up vectors, \ie with parastichy numbers $M,N$, the parastichy numbers of any of its chains are also equal to $M,N$, which are equal to the numbers of left and right parastichies. All the up and down vectors of the tiling are represented in the chain. A  chain must have $M+N$ primordia.}

\label{prop:parastnum}

\proof A left parastichy coming from above a chain $C$ must first intersect $C$  at the origin of one of its down vector, whereas a right parastichy first intersects $C$ at the origin of an up vector. This provides a one-to-one correspondence between left and right parastichies and down and up vectors in a chain respectively. Hence there are $M$ down and $N$ up vectors (for a total of $M+N$ primordia) in $C$. On the other hand, each point of the parastichy $LP_{(k-1, j)}$ is the origin of the down vector $\dd_k$, and likewise for $RP_{(i, l-1)}$ and $\uu_l$. This provides a 1-1 correspondence between parastichies and the vectors that originates at their points, and thus a 1-1 correspondence between the set of up and down vectors of the tiling and those of $C$.
\qed

\remark {\bf (Number of petals in daisies)} The predominance of flowers whose number of petals is a Fibonacci number was observed before people had noticed the relationship of these numbers to that sequence \cite{grew}. Modern studies also show that  the number of petals in \emph{asteracea} (\eg daisies) has a statistical peak at   Fibonacci numbers (see references in \cite{battjes}). Proposition \ref{prop:parastnum}  provides, among other things, a possible explanation as to why this might be.  If we accept that parastichy numbers of the inflorescence of these plants are predominently successive Fibonacci numbers, this phenomenon would simply be a consequence of the fact that the ray petals occur, statistically, at a single primordia front. A front has $N+M$ primordia, a Fibonacci number if $N$ and $M$ are successive Fibonacci numbers. It would be interesting to check experimentally the hypothesis of petals forming predominantly at a single front.

\theorem{In a fat rhombic tiling $\tiling$, there is a unique primordia front at each point.} 
\proof We describe the algorithm that builds the front $p_{i_1}, \ldots, p_{i_{M+N}}$. Given a point $p_{i_1}=z_{m,n}$ of $\tiling$, let $p_{i_2}= z_{m+1,n}$ be the right parent of $p_{i_1}$. By induction let $p_{i_{j+1}}$ be the right child of $p_{i_j}$ unless this child is strictly higher than $p_{i_1}$, in which case let $p_{i_{j+1}}$ be the right parent of $p_{i_j}$.   We now show that this process has an end. Suppose by contradiction that the piecewise linear curve $c_r$ we built crosses the vertical line through $p_{i_1}+(1,0)$ strictly below that point. Construct with a similar algorithm a curve $c_l$ starting from $p_{i_1}+(1,0)$, but going left. The curves $c_r$ and $c_l$ necessarily cross at a tiling point in the strip between the vertical lines through $p_{i_1}$ and $p_{i_1}+(1,0)$. Let $Q$ be the rightmost such crossing point. The right parent of $Q$ is in $c_r$. The right child of $Q$ is in $c_l$ and is thus lower than $p_{i_1}$, which contradicts the algorithm for $c_r$. The set of points obtained is clearly a chain. It is a front because a child of any of its primordia is higher than $p_{i_1}$, by construction. \qed

The periodicity proven in the following proposition is illustrated in Fig. \ref{fig:Periodpen}.

\proposition{ There are $MN$ primordia in a \ph~tiling $\tiling$ of parastichy numbers $M, N$, between a point $z\in \tiling$ and its translate $Z= z+\vec U$, including $Z$, (where $\vec U = \sum_{i=1}^N \uu_i$)  in the ontogenetic order.}
\proof Let   ${\cal F}_z$ and  ${\cal F}_Z$ be the fronts at $z$ and $Z$ respectively. Because of the periodicity, ${\cal F}_Z = {\cal F}_z + \vec U$. The segment of $\tiling$ comprised between $z$ and $Z$ (including $Z$) includes all the points between ${\cal F}_z$ and  ${\cal F}_Z$, as well as  ${\cal F}_Z$. The segment comprises all the segments of $N$ primordia of the left parastichies strictly above  ${\cal F}_z$. Since there are $M$ such parastichies, the number of points in the segment including $Z$ is $MN$. \qed
\label{prop:frontperiod}

\section{The Snow Dynamical System \snow}
\label{sec:snow}
\subsection{Definition of \snow}
\label{subsec:snowdef}

Remember that, in the introduction, we gave the following intuitive definition of the Snow model: given a configuration of disks of equal diameter $D$ on the cylinder, place a new one in the lowest position possible on top of the configuration, avoiding overlaps.  

To turn this intuitive definition into a mathematical one, we made the choice of considering configurations of constant, finite number of disks on the cylinder. We achieve this simply by removing the last primordium in our list at each iterate, making sure that there are enough primordia so that this removal does not have perverse, artificial effects. This allows us to use the framework of dynamical systems where the space of configurations is of constant dimension throughout the time evolution. We also require that the configurations be ordered by height, the highest being the first one - so that the lowest one is the one removed at each iterate. To decide on the location of the new disk (primordium), we slide a circle $y =h$ up the cylinder, starting at the top
 of the configuration, and at each height $h$, we check whether there is room to place a disk of radius $D$  without overlapping disks in the configuration. This checking is done via computing minimum distances to points of the configuration along that circle. When the test is positive, we add the new disk and erase the last disk in the list. We call the height $y=y_*$ at which there is first room to place a primordium at the edge of the meristem the \emph{threshold} value.  In dry mathematical terms, this translates into (see Fig. \ref{fig:MinDis} for an illustration):

\definition{({\bf Snow map}) Define the map \snow~on $K$-tuplets of points of the cylinder $\cyl$  by
$$\snowm (p_1, \ldots,p_K) = (P_1(p_1,\ldots,p_K), p_1, \ldots,p_{K-1}) $$
where each point $p_k$ is given by its angular and height  coordinates $(x_k, y_k)$, and
where the function $P_1$ determines the center of the new primordium in the following way.
Let $$ Dis_{y,k}(x) = dist{((x,y),p_k)},$$ where $dist$ is the usual euclidean distance on $\cyl$, let $$MinDis_y(x) = \min_{k\in \{1, \ldots, K\}} Dis_{y,k}(x),$$ and let 
$$y_*  = \min  \{ y\geq \max_k y_k \ \mid \max_x MinDis_y(x) = D\}.$$ 

Finally,  define $P_1(p_1,\ldots,p_k)$ to be the point $(x_*, y_*)\in \cyl$ at which this ``minimaximin" is attained. If it is attained at several possible values of $x$, choose the smallest of those $x$ in the interval $(-\frac12, \frac 12]$. \label{def:snow}}

\begin{figure}[h] 
   \centering
   \includegraphics[width=4 in]{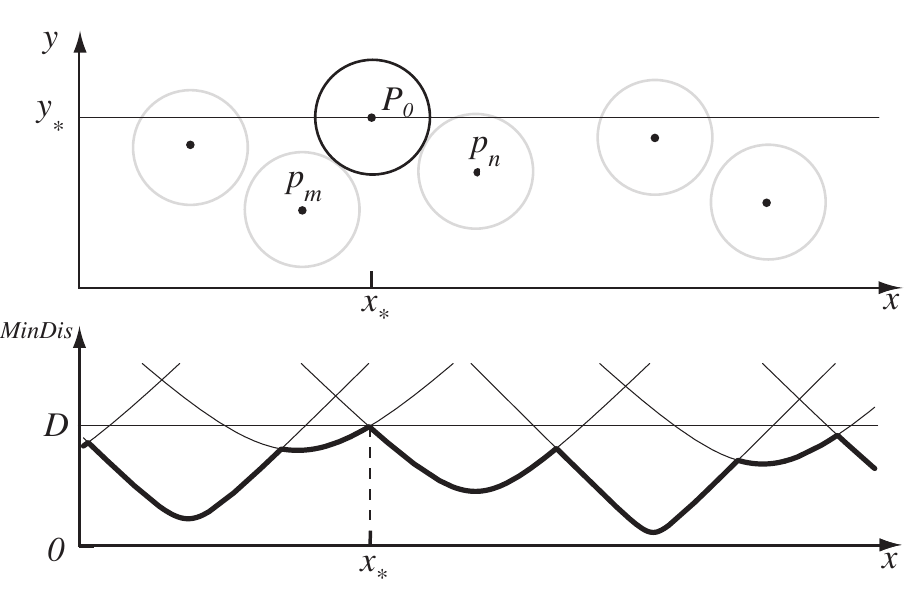} 
   \caption{\footnotesize {The function $MinDis_y$ at the threshold value $y_*$, shown below the corresponding configuration.  The maxima of $MinDis_y$ occur when the graphs (represented with lighter lines) of two convex functions $Dis_{y,k}$  cross. Thus at a maximum the corresponding test point $(x,y)$ is equidistant to its two nearest neighbors in the configuration. As the threshold value $y_*$ is reached where $\max_xMinDis_{y_*}(x) = D$, the point $(x_*, y_*)$ at which the maximum is attained is thus at distance $D$ from its two nearest neighbors $p_m$ and $p_n$: there is exactly enough space for a disk of radius $D$, and it must be tangent to the disks centered at $p_m$ and $p_n$. The convexity of the functions $Dis_{y,m}$  and $Dis_{y,n}$ also implies that $p_m$ and $p_n$ are on opposite sides of $P_1$. } }
   \label{fig:MinDis}
\end{figure}

\remark{Douady and Couder \cite{douadycouder}  used a similar idea in the algorithm for their Snow computer models (see also \cite{kunzthesis}). Instead of the $MinDis_y$ function, they used potentials which are the sum of ``repulsive" interactions with existing primordia, with interactions decaying as the distance  increases. The map \snow~can be seen as the limit of such models as their \emph{rate} of  decay goes to infinity: in the limit, the test primordium only ``feels" the closest primordium, as is the case in \snow. In \cite{douadycouder}, this limit is  called the ``hard disk" case.}

\remark{As shown in Figure \ref{fig:MinDis}, and its caption, the geometry of the function $MinDis_y$ implies the following:

\begin{enumerate}
\item {\it Equidistance and tangency to nearest neighbors. }  If  the two nearest neighbors of $P_1$ are less than $2D$ apart,  $P_1$ is tangent to them - and thus equidistant to them. If they are farther  apart, $P_1$ is located at their midpoint. We use the term of ``parents" for the two closest neighbors, even in the latter situation, generalizing the notion of Section \ref{subsec:parentsEtc}. Correspondingly, $P_1$ is the child of its parents.

\item {\it Opposedness of parents.} Generically (see Section \ref{subsec:differentiability}) , the new primordium has only two parents.  In this case the centers of these two anterior primordia must lay on opposite sides of a vertical line through the center of the new one. Hence, generically $P_1$ has a left and a right parent.
\end{enumerate}
}
\label{remark:Sgeom}

These properties form the basis of our computer algorithm, in which we draw lists of candidates new primordium by placing disks tangentially  to appropriate pairs of existing disks on sufficiently dense configurations. We then weed out the candidates that overlap with existing primordia or whose parents are not opposed and choose the lowest of the remaining candidates.  
In contrast to the above definition, our computer algorithm may add disks lower than the highest disk in the given configuration, when there is room for one - i.e. if the configuration has ``holes". For instance, the algorithm may fill in a necklace until it forms a front. Since all ``decent'' configurations eventually fill in and form a front at their top, we chose to elude, in this paper, the issue of which configurations eventually fill in and concentrate on configurations which already terminate by a front, or perturbations of such configurations.

\remark{\bf (A potential alternate definition of front)} In the case of a rhombic tiling, it is not hard to check that the primordia that contribute to the function $MinDis_y$ for a given $y$ form a front. Thus, for rhombic tilings, a front at a new born primordium can be defined as the set of primordia closest to the meristem. One could use this to generalize the definition of front to general configurations, letting go of the requirement of tangency.
\label{remark:frontclosest}

We will be specially interested in the lattices and tilings that are ``preserved" by \snow. To make this notion more precise:
\definition{{\bf (Dynamical Configurations)} An infinite configuration $\config$ is called \emph{dynamical} if given any of its segments $X$ of length $K$, $\snow(X)$ is the segment of $\config$ immediately above, \ie obtained by shifting the ontogenetic indices of $X$ by 1.}
\label{def:dynamicalconfig}

\subsection{Domain of Differentiability of \snow}
\label{subsec:differentiability}
We call a configuration in $\cyl^K$ a \emph{critical configuration} if two distinct pairs of parent primordia lead to two candidate children primordia at the same (lowest) level for the map \snow. This may occur when the function $MinDis_{y*}(x)$ attains its maximum at two distinct values of $x$. But it can also occur as a \emph{triple tangency} where two pairs of primordia sharing a common primordium give rise to the same child. In that case, the graphs of three functions $Dis_{y_*,k}$ forming $MinDis_{y_*}$ cross at the maximum.

\proposition{ The set of critical configurations is a closed set, finite union of manifolds (maybe  with boundaries) of codimension at least 1 (and thus of measure 0) in $\cyl^K$.}
\proof 
Given the location of one child candidate, one needs one (differentiable) equation to express the fact that another candidate belongs to the same horizontal line.  This equation can be written in the form $f({\config}) = 0$, where ${\config}$ denotes a configuration in $\cyl^K$, and $f$ is the algebraic function giving the difference of height of children of two distinct pairs of primordia in ${\config}$.  The gradient of $f$ is always non zero on the level set $f =0$. Indeed, let ${\config}$ be such that $f({\config})=0$ and let $p_L$ and $p_R$ be the parents of one of the two candidates. Choose an infinitesimal deformation  $\Delta \config$ of ${\config}$ that rotates $p_R$ around $p_L$, leaving all other primordia fixed. The corresponding  displacement of the candidate child is $\frac{\Delta \config }2$ since it rotates at half the radius. $f(\config +\Delta \config)$ is approximately the vertical component of $\frac{\Delta \config}2$, which is not 0, since $p_R$ and $p_L$ are not above one another. 
Hence $0$ is a regular value for $f$ 
and  the equation $f({\config}) = 0$ defines locally a manifold of codimension 1 (see the Preimage Theorem, \cite{difftopo}).

Each choice of  two distinct  ordered pairs of primordia gives rise to such a manifold. The number  of such critical manifolds is thus bounded by the choices of two distinct ordered pairs of distinct indices in $\{1, \ldots, K\}$. The critical set is closed: for each of the choices of pairs of (ordered) pairs of parents, the zero level set of the corresponding function $f$ is closed: at the points of $\cyl^K$ that a given pair of parents ceases to correspond to maxima, another pair must yield a maxima. Hence the limit of a sequence of critical configurations is always critical.
\qed

We call a configuration $\config$ \emph{ $q$-non-critical} if $\snowm^k(\config)$ is not critical for $k\in \{0, \ldots , q\}$ (we use $\snowm^0=Id$ here). We denote by  $NC_q$ the set of $q$-non-critical configuration. Note that  $NC_{q+1}\subset NC_q$.
 
\proposition {The map \snow~is continuous and differentiable on the open set $NC_0$ of non-critical configurations. More generally, the map $\snowm^{q+1}$ is continuous and differentiable on the open set  $NC_q$ of $q$-non-critical configurations.}
\label{prop:differentiable}
\proof $NC_0$ is the complement of a closed set and is thus open. Outside of the set of critical configurations,  the function $x\mapsto MinDis_y(x)$ has a unique maximum for each $y$ near the threshold value $y_*$. This maximum corresponds to two parent primordia, whose child $P_1$ is strictly the lowest candidate primordium. In a neighborhood of a non-critical configuration the parents indices of the new primordium do not change. $P_1$  is a differentiable, algebraic function of the two parents ($P_1$ is  the intersection of two circles centered at the parents).  So the first component function $P_1$ of \snow~in Definition \ref{def:snow} is continuously differentiable (and thus continuous) on non-critical configurations. All other component functions of \snow~are trivially continuously differentiable.

We show that the set $NC_q$ is open, and that $\snowm^{q+1}$ is differentiable on it,  by induction on $q$. We have proven the first step of the induction for $q=0$ above. Assume $NC_k$ is open and $\snowm^{k+1}$ is continuous on $NC_k$ for $k<q$. Let $\config \in NC_q$. This implies that $\snowm^q(\config)\in NC_0$ and that $\config \in  NC_{q-1}$. Since $\snowm^q$ is continuous on the open set $NC_{q-1}$ (by induction hypothesis) and $NC_0$ is open, we can find an open neighborhood $U(\config)\subset NC_{q-1}$ such that $\snowm^q(U(\config))\subset NC_0$. This implies that $U(\config)\subset NC_q$ which makes $NC_q$ open. $\snowm^q$ is differentiable on $NC_q$ since this set is a subset of $NC_{q-1}$. Since
$\snowm^q(NC_q)\subset NC_0$, and $\snowm$ is differentiable on $NC_0$, the composition $\snowm^{q+1}$ is differentiable on $NC_q$, by chain rule.  \qed\\

We call the \emph{orbit segment of length $q$} of a configuration $\config$ the configuration $X_q(\config)$ in $\cyl^q$ made of $q$ first points of $\snowm^q(\config)$.

\begin{corollary}{\bf (Continuity of Parents Data )} The function associating to a configuration $\config$ the parent data of $X_{q+1}(\config)$ is constant on each connected component of $NC_q$.
\end{corollary} 
\label{corollary:parentscontinuity}

\begin{proof}
Connected components of $NC_q$ are by definition the open sets of configurations whose $k^{th}$ new primordium under \snow~have the same parent indices, for $k\in \{1, \ldots, q+1\}$. \end{proof}

Discontinuities of \snow~do occur at configurations on the boundaries of the connected components of $NC_0$,  where there are multiple maxima for $MinDis_y$ at the threshold level $y = y_*$. Indeed, two different configurations arbitrarily close to such a critical one may yield a new primordium in drastically different positions, although in the long run two such configurations might look arbitrarily similar. This switch of ontogenetic order is what makes the divergence angle a less than adequate classifying tool for the geometry of configurations (see Section \ref{subsec:fibonacci}).

The map \snow~also fails to be differentiable at configurations with a triple tangency, for which $P_1$ has more than two equidistant nearest neighbors, even though it is continuous there. (This occurs for instance when the configuration is a segment of hexagonal lattice, corresponding to at a turning point of a branch of the bifurcation diagram of Figure \ref{fig:bifdiag}.) At those points, there is more than one choice for the differential matrix, violating differentiability.

\section{Dynamics and Geometry}
\label{sec:dyngeom}
In this section, we show that fronts determine the future of a configuration and the its changes of parastichy numbers. We also give some strong evidence that the set of dynamical tilings forms an attractor for \snow.

\subsection{Dynamical Properties of Primordia Fronts}
\label{subsec:frontdyn}
  We say that a front $\cal F$ is a \emph{top front} for a configuration $\config$  if $\cal F \subset \config$ and the points of $\config$ are in or below $\cal F$.  We based some of our fastest algorithms for \snow~on the following proposition:

\proposition{If a configuration $\config$ has a top front, so does $S({\config})$
and the top front of $S({\config})$ is the union of the new primordium $P_1$ and of the points of the top front of $\config$ that are not between the parents of $P_1$. If two configurations in $\cyl^K$ have the same top front, they have the same orbit segments.}

\proof
 If $\config$ has top front $\cal F$, the only functions $Dis_{y,k}$ whose graphs contribute parts to the graph of $x\mapsto MinDis_y(x)$ in the definition of \snow~are those corresponding to primordia in $\cal F$ (Remark \ref{remark:frontclosest}). Hence $\cal F$ determines the new primordium $P_1$ in the iteration. One can check that the union of $P_1$ and of all the primordia in $\cal F$ except for those between (if any) the left parent and right parent of $P_1$ (in the front ordering) constitute a top front for $S({\config})$. By induction all the subsequent primordia in the orbit are determined by $\cal F$. \qed\\

The following proposition is at once simple and we think fundamental to understand phyllotactic transitions - pointing to the central role of the local geometry of fronts. The proof is essentially by picture (See Fig. \ref{fig:transitions}).

\proposition{Given a configuration $\config$ with a top front $\cal F$, the top front $\cal F'$ of  $S({\config})$ has same parastichy numbers as $\cal F$ if and only if the left and right parents of the new primordium $P_1$ are separated by exactly one primordium in $\cal F$. In this case $P_1$ creates a rhombic tile with $\cal F$. If the parents of $P_1$ are adjacent in $\cal F$, one of the parastichy numbers of $\cal F$ increases by one, and $P_1$ forms a triangular tile with $\cal F$. Finally, if the parents are separated by two primordia in $\cal F$, one of the parastichy numbers decreases by one, and the new tile is pentagonal.}
\label{prop:fronttransitions}
\begin{figure}[h] 
   \centering
   \includegraphics[width=6in]{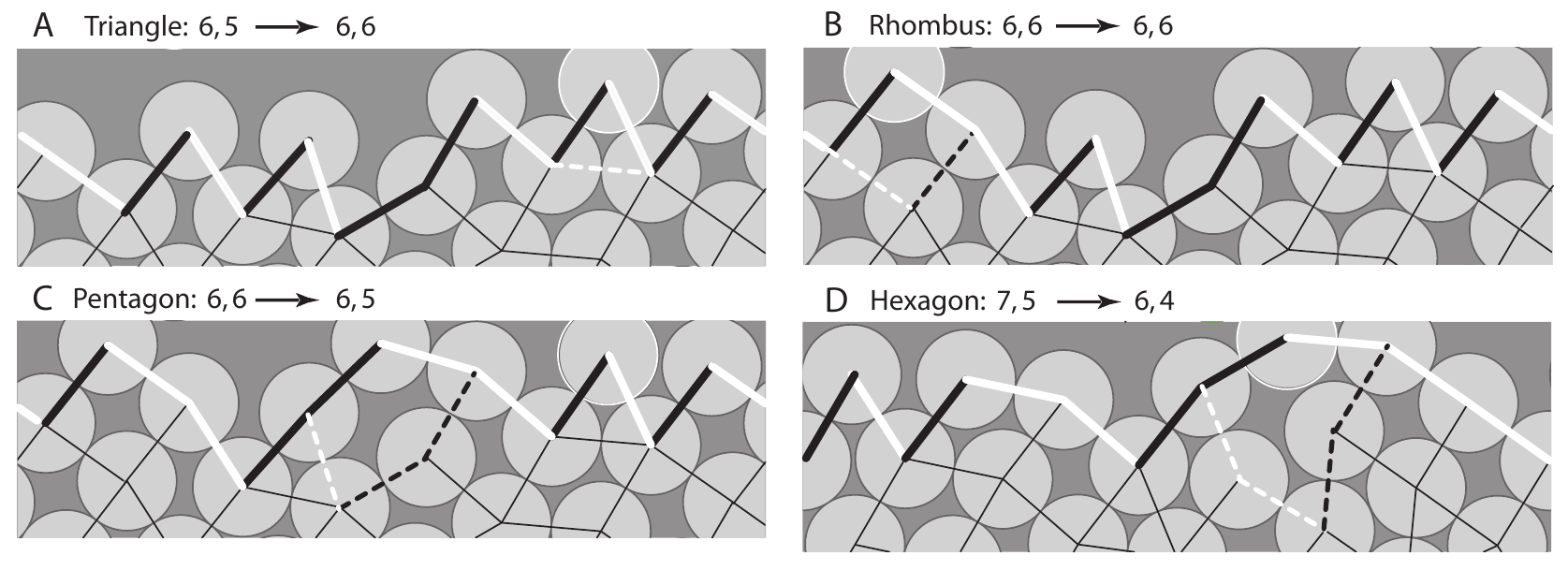} 
   \caption{\footnotesize{\bf The different front transitions. }{\sf (A)} A triangle transition. One down vector is replaced by another down (angle $-\pi/3$ with the original), and an up vector (angle $\pi/3$ with the down vector replaced). One of the front parastichy numbers, the up number, increases by one, while the down number stays the same. {\sf (B)} A rhombus transition. A pair of down and up vectors has switched order in the front, with no change in the sets of down and up vectors. {\sf (C)} A pentagon transition. One down vector and a pair of up vectors are replaced by one up and one down. Hence  the front up parastichy number decreases by one. Note that there are no simple relationship between the angles of the new vectors and the old ones they replace.  {\sf (D)} A (much rarer) hexagon transition. A pair of down and a pair of up vectors are replaced by one vector each. These four orbit segments are taken from the same orbit, with A, B, C corresponding to successive iterations, and D an anterior one. }
  \label{fig:transitions}
\end{figure}

\proof  The transitions are understood by the number of down and up vectors between the left and right parents of the new primordium. These are always replaced by a pair of up and down vectors, in that order. The numbers of up and down vectors replaced determine the shape of the new tile of the ontogenetic graph, and the change of front parastichy numbers. See Fig. \ref{fig:transitions}. \qed

In our numerical experiments, the rhombic transitions are by far the most common, followed by triangles and pentagons - equally common when $\diam$ is constant, as they usually come in pairs. Hexagons are much rarer.

\conjecture{Configurations with top fronts cannot yield polygonal tiles with more than 6 sides.}
\subsection{Fixed Points and Periodic Orbits in the Shape Space} 
\label{subsec:fp.po}
We now consider the shapes of configurations that are preserved under some iteration of \snow. Configurations whose shape is preserved under \emph{any} iteration (fixed points) are found to be the same as for the Hofmeister map $\phi$ of \cite{jns}. In other words (see Definition \ref{def:dynamicalconfig}), the dynamical lattices of \snow~and $\phi$ coincide, for appropriate choices of parameters. On the other hand, we will see that many dynamical tilings for \snow~are not dynamical for $\phi$.

We first introduce a parameterization of the shape space of configurations, and the map \qsnow~that \snow~induces on it. The shape of a configuration is determined by its relative coordinates:
$$\ovv p_k =p_{k+1}-p_{k}, \qquad k \in \{1, \ldots, K-1\}.$$ 
This set of coordinates can be seen as a parameterization of the quotient space of the set of cylindrical configurations modulo the translations on the cylinder. As a particular example, a helical lattice in this quotient space is simply given by the equations $\ovv p_k=p_*$ for all $k\in \{1, \ldots, K-1\}$ and for a fixed $p_*\in \cyl$. 

The map \snow~induces a map \qsnow~on this quotient space, of the form   $$\qsnowm(\ovv p_1 \ldots, \ovv p_{K-1}) )=(\ovv P_1, \ldots, \ovv P_K)= (\ovv P_1(\ovv p_2 \ldots, \ovv p_{K-1}),\ovv p_2 \ldots, \ovv p_{K-2})). $$ 

\emph{Since $K$ above is an arbitrary large integer, we set $K = K-1$ for a lighter notation in the rest of this section.}

Similarly to \snow, in the Hofmeister map $\phi$  the placement of the new primordium is determined by the maxima of the function $MinDis$ (called $D$ in \cite{jns}), but instead of being at a threshold level $y = y_*$, it is evaluated at fixed, equal intervals of $y$. The interval length, called internodal distance and denoted by $y$ in \cite{jns}, is the parameter for that system.   In \cite{jns}, inspired by \cite{leelevitov}, we used hyperbolic geometry to analyse in detail the fixed points set of the Hofmeister map for all values of $y$. It turned out to be a subset of the set of (segments of) fat rhombic lattices, see Figure \ref{fig:bifdiag}. This latter set, described by  van Iterson \cite{vaniterson} must be truncated along crucial segments of its branches to obtain lattices that are dynamical for $\phi$. The same diagram was obtained by Douady \cite{douady}  in a geometric context which is essentially that of this present paper, using Euclidean geometry only. The next proposition shows that the fixed points sets of the maps $\phi$ and \qsnow~are identical when considering all values of the parameters $y$ and $D$.

\proposition{Fixed points for the map \qsnow~are segments of fat  rhombic lattices in $\cyl^{K}$. These fixed points are the same as for the Hofmeister map $\phi$ of \cite{jns} and their set can be visualized in the truncated van Iterson diagram of Figure \ref{fig:bifdiag}. }

\proof A fixed point for \qsnow~is such that $(\ovv P_1, \ldots, \ovv P_K)=(\ovv p_1 \ldots, \ovv p_K)$.
On the other hand,  the definition of \snow~gives
$\ovv P_k=\ovv p_{k-1},  \ k\in \{2,\ldots,K\}$. Hence, $\ovv p_k =\ovv p_{k-1},  \ k\in \{2,\ldots,K\}$. It easy to see that this yields, in the absolute coordinates $p_k= p_0+kp_*$ for some $p_0$ and $p_*$  independent of $k$,    proving that fixed points of \qsnow~are segments of helical lattices. By Remark \ref{remark:Sgeom}, these lattices must be rhombic and opposed. If $\config$ is fixed for \qsnow,  we just saw it is a segment of lattice, and by periodicity of the lattice, it has constant internodal distance $y$ between successive points in its ontogenetic order. Since the new primordium maximizes $MinDis$ at its level, it must correspond to the choice of new primordium for the map $\phi$, for that value of the parameter $y$. Thus $\config$ is fixed under $\phi$, for that value of parameter $y$. Conversely, if $\config$ is fixed under $\phi$, we showed in \cite{jns} that it is a segment of (fat) opposed rhombic lattice. Let $D$ be the mutual distance of points in this lattice. The new primordium of $\config$ under $\phi$ maximizes $MinDis$ at the threshold values corresponding to $D$ and is thus the new primordium for \snow, proving $\config$ is fixed for \qsnow.   \qed
 
\remark In the bifurcation diagram, there is a monotone correspondence $y\mapsto D(y)$ between  the parameter $D$ for the map \snow~and the internodal distance 
$y$ used as parameter for $\phi$ (see \cite{douadycouder}).  When the $M,N$ branch is not truncated (the so called ``regular case", where $M<2N$ and $N<2M$, see \cite{jns}) the angle between the up vector and the down vector of the lattice spans the range of $[\frac{\pi}3, \frac{2\pi}3]$. Simple trigonometry on a necklace of the lattice shows that this corresponds to  the parameter $D$ ranging in $\left[({M^2+N^2+MN})^{-\frac12}, ({M^2+N^2-MN})^{-\frac12}\right]$.
\label{remark:Drange}

\begin{figure}[h] 
   \centering
   \includegraphics[width= 5.5in]{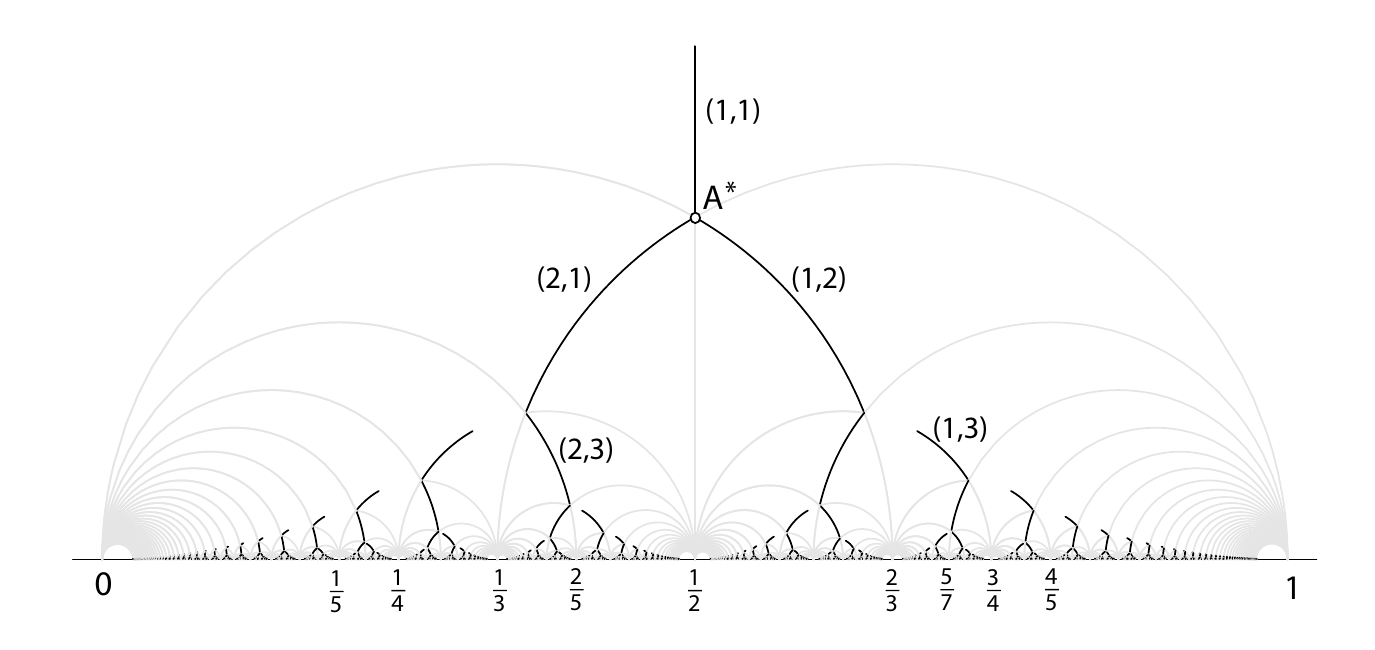} 
   \caption{\footnotesize{Fixed point set for the Snow and Hofmeister maps \qsnow~and $\phi$. Each point $(x,y)$ in this plane represents the generator of a cylindrical lattice. The lattices that correspond to fixed points for \qsnow~and $\phi$ have their generator along the dark arcs of circle - each dark point representing a dynamical lattice. We have indicated a few parastichy number pairs  corresponding to different branches. The grey arcs bound regions of constant parastichy numbers. In \cite{jns} , the coordinate $y$ of the generator is used as a parameter for the map $\phi$, and this graph is the fixed point bifurcation diagram. A monotonic change of coordinates $y \mapsto D(y)$ would give the topologically equivalent bifurcation diagram for \qsnow.}}
   \label{fig:bifdiag}
\end{figure}

\proposition{Periodic points are segments of multilattices. If an orbit is a segment of a \ph~tiling of parastichy number $M,N$, then the tiling is fat and rhombic and the orbit has period $MN$ for \qsnow.}
\proof    Periodic orbits of period $q$ are such that $\qsnowm^q(\ovv p_1 \ldots, \ovv p_K)=(\ovv p_1 \ldots, \ovv p_K)$. The same argument as above implies that $\ovv p_{k+q}= \ovv p_k$.  This makes the configuration a multilattice with generator $\sum_{j=1}^q  \ovv p_j$ and translation vectors $v_k= \sum_{j=1}^k \ovv p_j, \ k\in \{1, \ldots, q\}$. The fact that a tiling-orbit is fat and rhombic is an immediate consequence of the interpretation of \snow~as a process of piling non overlapping disks of the same size in $\cyl$. Opposedness comes from the optimization involved in the definition (see comments at the end of Section \ref{subsec:snowdef}).  The corresponding orbit of \qsnow~is periodic of period $MN$ since, by Proposition \ref{prop:frontperiod}, each front's shape is repeated every $MN$ iterates. \qed

\conjecture{Periodic points for \qsnow~are segments of fat rhombic tilings.}

The proof of this conjecture would rest on the fact (to be established) that no periodic orbit may contain other tiles than rhombi (apart from rhombi at the boundary of fatness which can be interpreted as two triangles).

We now obtain, with relatively little work, infinitely many sets of periodic orbits for the map \qsnow~that could not exist for the Hofmeister map, since two primordia could not be generated at the same height in that model. The cover of a $k$-jugate configuration $L_k$ is obtained from a lattice $L$ by gluing $k$ copies of the cover of $L$, rescaled by $1/k$ (see Fig. \ref{fig:Lattice&Whorls}). Since the cover   $\tilde L_k$ is homothetic to  the cover $\tilde L$, and homothecies preserve angles and equidistance, $L$ is respectively rhombic, opposed, or  fat if and only if $L_k$ is. We now show that the the correspondence $L\mapsto L_k$ maps fixed points of \qsnow~to periodic points of period $k$.

\proposition{ A  lattice $L$ is dynamical if and only if its corresponding $k$-jugate configuration $L_k$ is. A segment of a dynamical $L_k$ is a periodic point for \qsnow, of period $k$.}

\proof  Choose a segment $\config_k$ of $L_k\in \cyl^{kK}$ which has $k$ primordia at the same top level.  The segment $\config_k$, as a subset of $\cyl$, can be seen as $k$ rescaled copies of a segment  $\config$ of $L$ in $\cyl^K$ (see Fig. \ref{fig:Lattice&Whorls}) set side by side on the cylinder. We can choose $\config$ and $\config_k$ to have the same base point $(0,0)$. Accordingly, the graph of the function $MinDis_{y/k}$ in the definition of \snow~for $\config_k$  is made of $k$ copies set side by side, rescaled by $1/k$, of the graph of  $MinDis_{y}$ for $\config$.  The value $y$ is a threshold for the function $MinDis_y$ for $\config$ (given the parameter $\diam$) at a point of $L$ if and only if $y/k$ is a threshold for $\config_k$ (with parameter $D/k$) at $k$ points (on the same level) of $L_k$. This implies that $L_k$ is dynamical if and only if $L$ is dynamical.  Since a segment of $L_k$ has from 1 up to  $k$ primordia at each level, an orbit of a segment of $L_k$ is of period $k$: a segment of  $L_k$ in $\cyl^{kK}$ can translate into another one if and only if they both have the same number of primordia at the top level.    \qed

\subsection{Sufficient Conditions for Periodicity}

This Section provides a useful and easily implemented test to establish in a finite number of iterates of \snow, whether an orbit is part of a tiling, and thus periodic.

We call an \emph{ancestor} of a primordium $p$  in a   \ph~tiling $\tiling$ a primordium $A$ which can be joined to  $p$ by a connected sequence of up and negative down vectors of $\tiling$. In a rhombic tiling, this is equivalent to the intuitive meaning of ancestor (parent of parent of~...). It is not hard to see that in this case, $p$ is on or above the necklace formed by the left and right parastichies between the ancestor $A$ and $A+\vec U$.  The function $\lceil x\rceil$ used below denotes the smallest integer greater than $x$. 

 \theorem{Let $\config$ be a configuration with top front. If the top fronts of $S^q(\config)$ for $q$ in $\{0, \ldots, \lceil\frac{3MN}2\rceil\}$  have all the same  parastichy numbers $(M,N)$, then  any orbit segment of $\config$ is the segment of the same dynamical tiling.  In particular, the \qsnow~-orbit of $\config$ is periodic, of period $MN$. }
 \label{thm:parastnum}\\
 \proof  Consider the rhombic tiling $\tiling$ generated by the left and down vectors of the top front $F_0$ of $\config$.
We will show that $\tiling$ is in fact a dynamical tiling, and that points of $\tiling$ above $F_0$  form the \snow-orbit for $\config$.  Since the front parastichy numbers are constant, all transitions in the orbit are rhombic (see Proposition  \ref{prop:fronttransitions}). 
Assume by induction that the top front $F_k$ (with  $k< \frac  {3MN}2-1$) of $\snowm^k(\config)$ is in $\tiling$ (by hypothesis $F_0$ \emph{is} in $\tiling$). Let $P$ be the new primordium at iteration $k+1$, and $z_{m,n}$ be its left parent in $F_k$. Since the transition is rhombic, $z_{m+1, n+1}$ must be its right parent and $P = z_{m,n}+\uu_{n+1} = z_{m,n+1}\in \tiling$. Thus any orbit segment of $\config$ is a subset of $\tiling$.
\label{theorem:conditionperiodicity}

We will now show that any orbit segment $X$ of $\config$  is  in fact a full segment of the tiling $\tiling$ above $F_0$.  Suppose first that the orbit segment of length $MN$ above $F_0$ is equal to a segment of $\tiling$. By Proposition \ref{prop:frontperiod}, the front at iterate $NM$ is a translate of $F_0$ and thus the orbit shape is periodic of period $MN$ for \qsnow~and must coincide with $\tiling$ above $F_0$.  We now show that this is the only case possible.

Assume by contradiction that some point $z$ of the segment of $\tiling$  of length $MN$ above $F_0$ is \emph{not} in the orbit segment $X$ of same length, and choose $z$ to be the lowest such point above $F_0$. The orbit is then strictly bounded above by the necklace formed by the left and right parastichies between the points $z$ and $z+\vec U$ of $\tiling$. Indeed, since any point in the orbit segment $X$ is in $\tiling$, if a point of $X$ were above the necklace, it would have $z$ as an ancestor, which is absurd since $z$ is not part of the orbit. But, as is not hard to check, the number of points of $\tiling$ comprised between the front immediately below $z$ and the necklace is strictly less than $\lceil\frac{MN}2\rceil$, and thus the number of points between $F_0$ and the necklace is strictly less than $MN+\lceil\frac{MN}2\rceil = \lceil\frac{3MN}2\rceil$. This is a contradiction to the fact that at least  $\lceil\frac{3MN}2\rceil $ points of the orbit are in $\tiling$. \qed

\subsection{Attracting Manifolds of Periodic Points}
\label{subsec:attractor}
The following theorem shows that, around each helical lattice segment of parastichy numbers $(M,N)$ of the bifurcation diagram, apart for its turning points, there exists a superattracting manifold of dimension $M+N$ of dynamical tilings (periodic of period $MN$ for \qsnow) on which neighboring orbits land in finite time. In fact,  such a manifold exists near any sufficiently non-critical segment of dynamical tiling. We remind the reader that $NC_q$ is the open set of $q$-non critical configurations, and that $\rtmndk$ is the manifold of segments of length $k$ of rhombic tilings of parastichy numbers $(M,N)$ and parameter $\diam$.
 
 \begin{theorem}
Let $M$ and $N$ be coprime and let $K\geq M+N$. The set  $\dmndk$ of segments  of dynamical tilings of length $K$ in $NC_{MN}$ is an open submanifold of $\rtmndk$. 
Moreover,  there exists an open neighborhood $\cal V$ of $\dmndk$ in $\cyl^K$  such that, for  any configuration $\config\in \cal V$, and any $j>K$, $S^j(\config)$ is in $\dmndk$. The manifold $\dmndk$ is not empty  when a non-critical $M,N$-lattice of the bifurcation diagram exists for the given parameter $D$. 
\label{thm:attractor}  
\end{theorem}
\begin{proof}   
Take a segment $X_K(\tiling)$ of a dynamical tiling $\tiling$ and assume $X_K(\tiling)$ is in  $NC_{MN}$. Then it   is in fact in $NC_q$ for all $q\geq 0$,  by periodicity (Theorem \ref{theorem:RTperiodic}). Let $O$ be the (open) connected component of  $NC_Q$ containing $X_K(\tiling)$, for a chosen $Q\geq  \lceil\frac{3MN}2\rceil$. By Proposition \ref{prop:dimension}, $X_K(\tiling)$ is also contained in an open subset $U$ of the manifold $\rtmndk$: the boundary of the set of non fat tilings is made of critical tilings, thus, since $X_K(\tiling)$ is in $NC_q$ it is also strictly non-fat.  The set ${\cal O} = U\cap O$ is thus an open submanifold of tilings segments in $\rtmndk$, containing $X_K(\tiling)$. Let $Y$ in ${\cal O}$ be a segment of a tiling $\tiling'$. Since $Y$ is in the same component of $NC_Q$ as $X_K(\tiling)$,  parastichy numbers of the successive top fronts of $\snowm^j(Y)$  are constant  for $0\leq j\leq Q$. Since all the transitions are rhombic, these fronts are all fronts of $\tiling'$.  By Theorem \ref{theorem:conditionperiodicity}, $\tiling'$  is dynamical.  We have shown that $\dmndk$ is an open submanifold of $\rtmndk$. 

 Given $X_K(\tiling)\in \dmndk$, take a  configuration $\config$ in $\cyl^K$ in the same (open) connected component of  $NC_Q$ as $X_K(\tiling)$ ($\config$ need not be a tiling). Since all iterates of $X_K(\tiling)$ have a top front of parastichy numbers $(M,N)$, $\snowm^j(\config)$ must also have a top front of parastichy numbers $(M,N)$ for $k\leq j \leq Q.$  By Theorem \ref{theorem:conditionperiodicity}, $\snowm^j(\config)$ is thus a segment of a dynamical tiling of parastichy numbers $(M,N)$. The union $\cal V$ of all the $NC_Q$-connected  components of configurations in $\dmndk$ is open and attracted to $\dmndk$ in finite time.
 \end{proof}

In the regular case ($M<2N$ and $N<2M$, see Remark \ref{remark:Drange}), each fixed point helical lattice in the range  $D\in\left(({M^2+N^2+MN})^{-\frac12}, ({M^2+N^2-MN})^{-\frac12}\right)$ is non-critical, and by periodicity $q$-non-critical for all $q\geq 0$. Thus for each $D$ in this range, the manifold $\dmndk$ is non-empty, and provides a manifold of dimension $M+N$ of periodic orbits. In the irregular case, the range of allowable $D$ is smaller, but not empty. Since $k$-jugate lattices have $k$ primordia at the same level, they are automatically critical. Nonetheless, a perturbation argument should show:

\conjecture{Around any dynamical $k$-jugate lattice of parastichy numbers $M, N$  whose corresponding helical lattice is non-critical, there is an open set of dynamical tilings in 
$\rtmndk$.}

\conjecture{The set of dynamical tilings of given parameter $\diam$ forms an attracting invariant branched manifold - with branches of different dimensions - for the map \snow~which, in nearby systems, persists as an attracting invariant nearby (branched) manifold.}

Confirming Theorem \ref{thm:attractor}, a computation shows that the characteristic polynomial for the differential of \snow~at a non-critical dynamical lattice of parastichy numbers $(M,N)$ has the neat form $Char(\lambda) = \lambda^B(1-\lambda^M)(1 - \lambda^N)$, where $B= K-M-N$ ($K$ is the dimension of the phase space).  Thus the eigenvalues are either 0 or equal to $M^{th}$ or $N^{th}$ roots of unity. Clearly the dimensions of the generalized eigenspaces corresponding to the roots of unity sum up to $M+N$ which shows that the sum of these spaces must equal the tangent space to $\dmndk$ at $\lattice$. The zero eigenvalue in the complementary subspace shows that $\dmndk$ is normally hyperbolic at $\lattice$. Numerical evidence indicates the same to be true at dynamical tilings. Geometrically, the super attraction correspond to fronts forming in finite time on configurations that might not have them. 

Thus the set of dynamical tiling could entirely be made of pieces of normally hyperbolic invariant manifolds.  There are theorems (\eg \cite{fenichel}) that show that, given certain conditions on the map or flow, normally hyperbolic invariant manifolds survive perturbations of the system, as perturbed invariant manifolds. In the case of the map \snow, we have some hurdles stacked against us: 1) \snow~is not a diffeomorphism (it is not 1-1);  2) \snow~is not continuous everywhere;  3) the set of tilings has branches of various dimensions, some of which connect.

 \section{A Glimpse at the Set of Dynamical Tilings}
\label{subsec:rp2}
We conclude this paper with a numerical study of the topology of dynamical tilings.
To see how dynamical tilings of different parastichy numbers coexist, we look at the shape space of all chains of four primordia of diameter $\diam=0.3$.  Such a chain is given by four points, or vectors   $\vec v_1, \vec v_2,\vec v_3$ and~$\vec v_4$. We fix one primordium at the origin, $\vec v_1=(0,0)$, which has no consequence on the shape of the chain. We choose two angles, $\alpha$ and~$\beta$, as parameters as   Fig.~\ref{fig:parameterization} shows. A choice of these angles gives primordia located at $\vec v_1,\vec v_2$ and~$\vec v_3$ as follows
\begin{eqnarray*}
\vec v_1&=&(0,0),\\ 
\vec v_2 &=& (\diam\cos{\alpha},\diam \sin{\alpha}), \\  
\vec v_3 &=& \vec v_2+ (\diam\cos{\beta},\diam\sin{\beta})=(\diam(\cos\alpha+\cos\beta),\diam(\sin\alpha+\sin\beta)).
\end{eqnarray*}

 \begin{figure}[h] 
   \centering
   \includegraphics[height=1.3 in]{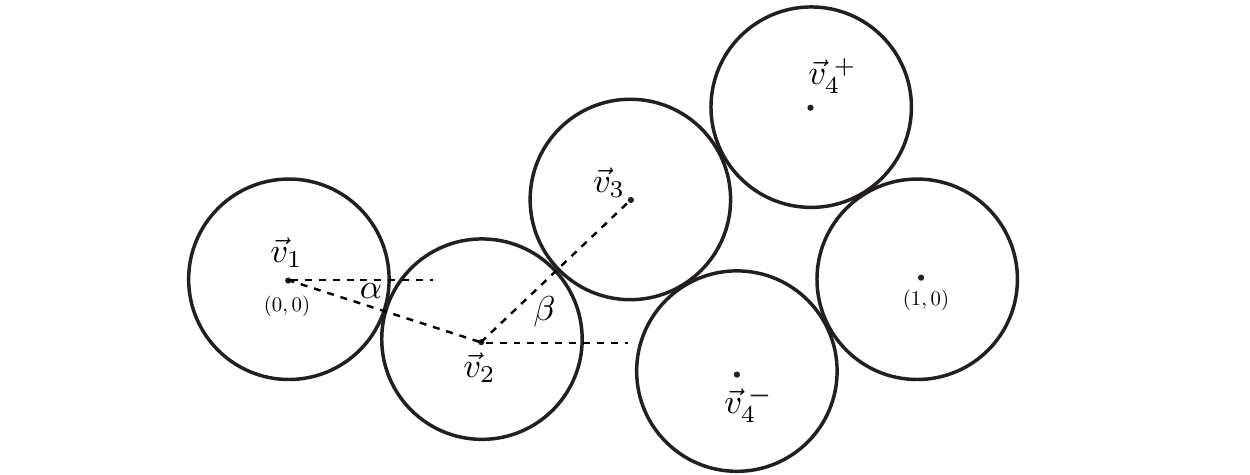} 
   \caption{\footnotesize {Parameterization of   4-chains.            }  }
   \label{fig:parameterization}
\end{figure}

To be able to define $\vec v_4$ and effectively have a chain with four primordia, we must have
$$\diam\le \textrm {distance}(\vec v_3,(1,0))\le 2\diam,$$
which is equivalent to
$$ 1\le \left(\cos\alpha+\cos\beta-\frac 1 \diam\right)^2+(\sin\alpha+\sin\beta)^2\le 4.$$

 \begin{figure}[h] 
   \centering
   \includegraphics[height=3.3 in]{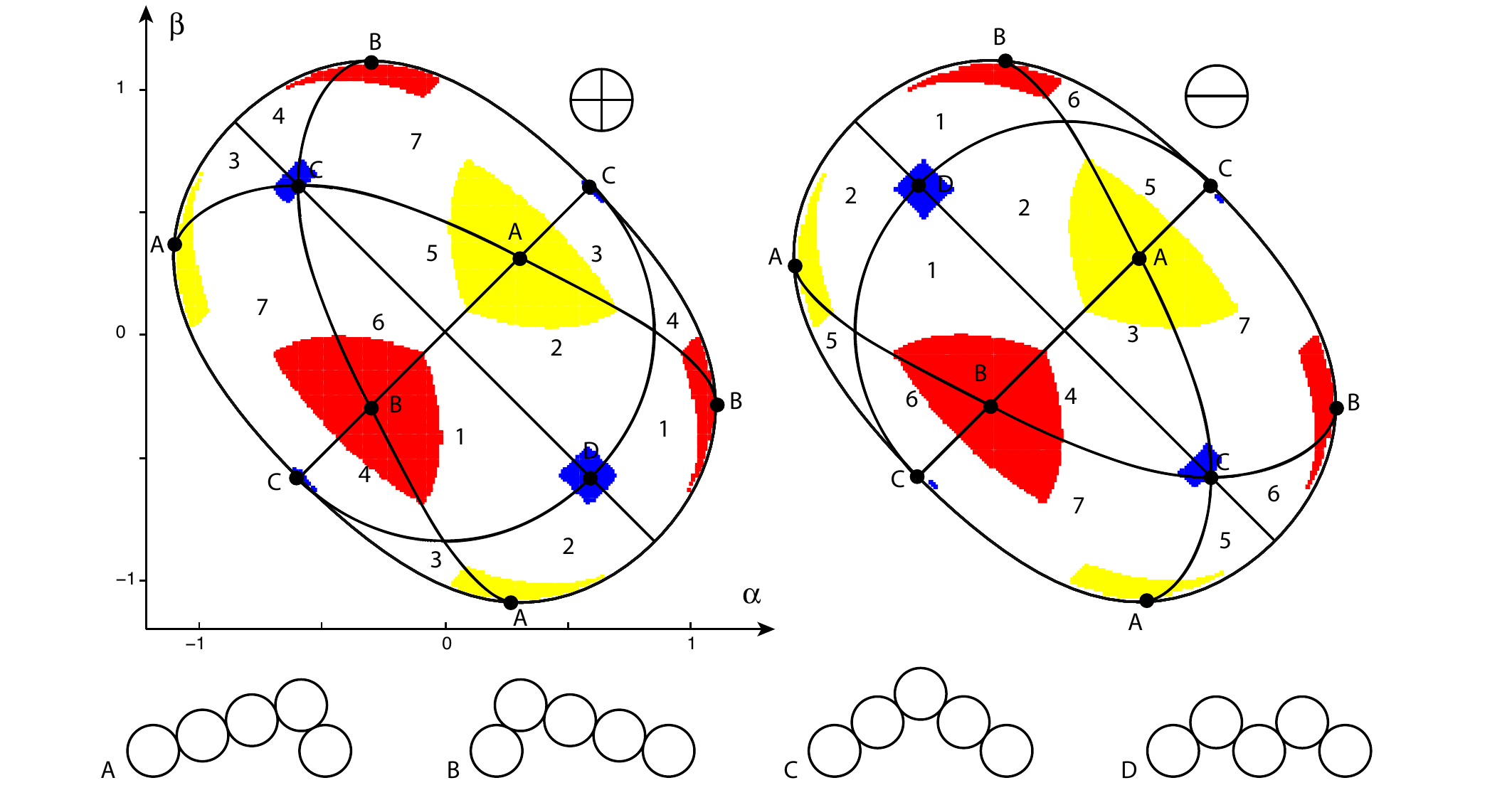} 
   \caption{\footnotesize {The shape space of chains of 4 primordia of diameter $\diam=0.3$. The  regions $\oplus$ and $\ominus$ are identified at their boundary (``equator"), forming a sphere. The black curves drawn are the images of the equator under the cyclic permutations of primordia in the chain. These curves separate regions,  which are identified when representing chains of same shape, according to the numbers shown. The union of the regions 1-7 forms a topological disk with antipodal identification at the boundary, yielding a space topologically equivalent to the projective plane~$\mathbb P^2$.  The point $A$ corresponds to the front of a $(1,3)$ lattice, $B$ to that of the $(3,1)$ lattice, $C$ and $D$ correspond to the two possible fronts of the bi-jugate (2,2) configuration. All these lattices are dynamical. The yellow region corresponds to chains that are fronts of dynamical tilings of parastichy numbers (1,3), the red region corresponds to fronts of (3,1) tilings and the blue one corresponds to fronts of (2,2) tilings.   Other chains are either not front and/or transit in one iterate of \snow~to the above colored regions, or to (2,3) and (3,2) tilings. 
 } }
   \label{fig:partition}
\end{figure}

For our choice of parameter $\diam$, the first inequality is always satisfied. In the case of strict (second) inequality, for each  choice of $\alpha$ and~$\beta$ the placement of primordia $\vec v_4$ is determined up to two possibilities. These are depicted in  Fig.~\ref{fig:parameterization},  labeled as $\vec v_4^{\ +}$ and~$\vec v_4^{\ -}$. So, without the further identifications that we will make below, the parameterized shape space of 4-chains consists of two copies $\oplus, \ominus$ of the same disk-like region, depicted in~Fig.~\ref{fig:partition}. The boundary points of these regions (the \lq\lq equator") correspond to choices of $\alpha$ and~$\beta$ for which the distance between primordia $\vec v_3$ and $\vec v_1=(0,0)=(1,0) \textrm{ mod }1$ is~$2\diam$, and so $\vec v_4^{\ +}=\vec v_4^{\  -}$. Hence, the boundary points of the  region $\ominus$ are identified one by one to those of the region $\oplus$, making the set a topological sphere (for now).

We further identify the four configurations $(\vec v_1, \vec v_2,\vec v_3,\vec v_4)$, $(\vec v_2,\vec v_3,\vec v_4,\vec v_1)$, $(\vec v_3, \vec v_4,\vec v_1,\vec v_2)$ and $(\vec v_4, \vec v_1,\vec v_2,\vec v_3)$, as they have the same shape on the cylinder. This leads us to make some identifications in the above sphere. Fig.~\ref{fig:partition} shows these with a number coding: regions with same number are identified. One can see that the shape space for 4-chains is topologically equivalent to the projective plane~$\mathbb {RP}^2$. The colored regions  correspond to shapes of \emph{fronts} of dynamical tilings of parastichy numbers (1,3), (2,2) and (3,1). This  coloring was obtained numerically by  sweeping the shape space, checking for front conditions, parastichy numbers and rhombic transitions for each chain in a grid of about 10000 points. Only 6 iterates of \snow~were necessary at each chain, thanks to Theorem \ref{theorem:conditionperiodicity}. Note that, after identification, the yellow and red regions each have only one connected component, on which the map has period $3 =  1\times 3= 3\times 1$. On the other hand, the (blue) shape space of fronts of (2,2) dynamical tilings   is disconnected, and the map \qsnow~toggles between one component and the other at each iterate, with period $4 = 2\times 2$. If one were looking at the shape space of \emph{dynamical tilings} (and not their \emph{fronts}), one would have to identify points in each orbit of \qsnow, in each colored region. The space of (2,2) dynamical tilings is then apparently connected: the fronts $\sf C$ and $\sf D$ are identified as belonging to the same cylindrical lattice, for instance.

\conjecture{ The shape space of dynamical tilings of parastichy number $M,N$, and parameter $\diam$ is contractible.}

{\bf Acknowledgments.} Even though  the bulk of the research for this paper occurred before our collaborations on this model with St\'ephane Douady, Jacques Dumais and Scott Hotton, it was solidified by many invaluable discussions with them. Dumais helped us see what might be useful to biologists and suggested our fitting of cylindrical plant patterns to tilings. Hotton made helpful suggestions on the programming used  in Figure \ref{fig:tilingfit}. Douady encouraged us to pursue the study of rhombic tilings and their fronts, showed us enticing experiments about them, and pointed us to the work of van Iterson on zickzacklinie.  Hotton, Dumais, and Luke Grecki read a draft of this paper and made many good suggestions. Jordan Crouser, Anna Naito, Duc Nguyen and Erich Kummerfeld helped gather plant data that was used in Figure \ref{fig:tilingfit}. We thank them all warmly.
This research was supported by the NSF/NIH collaborative  research  grant \# 0540740 and a Mellon collaborative grant.


\begin{thebibliography}{999}
\itemsep -3pt


\bibitem{adler}
Adler I. ``A Model of Contact Pressure in Phyllotaxis." {\it J. Theor. Biol.}
{\bf 45} (1974):1-79

\bibitem{jns}
Atela, P.,  Gol\'e, C. \&  Hotton, S. ``A Dynamical System for Plant Pattern Formation: Rigorous Analysis."  {\it J. Nonlinear Sci.} {\bf 12}:6 (2002):641-676

\bibitem{phylloweb} Atela, P., Gol\'e and Smith students,  {\it Phyllotaxis: an interactive site for the mathematical study of plant pattern formation} http://www.math.smith.edu/~phyllo

\bibitem{auxintraas}
Barbier de Reuille, P. , Bohn-Courseau,I.,  Ljung, K., Morin, H.,  Carraro, N.,  Godin, C.,  and Traas, J.: 
``Computer simulations reveal properties of the cell-cell signaling network at the shoot apex in Arabidopsis."
{\it  Proc. of Natl. Acad. Sci.}, {\bf 103} (2006):1627-1632 


\bibitem{battjes} Battjes, J. \& Prunsinkiewicz, P. ``Modelling Meristic Characters of Asteracean Flowerheads." in  {\it Symmetry in Plants}
World Scientific, (1998):  281-312.



\bibitem{douady}
Douady S. ``The Selection of Phyllotactic Patterns." {\it Symmetry in Plants}
World Scientific, (1998):  335-358.

\bibitem{douadycouder}
Douady S. \& Couder Y.  ``Phyllotaxis as a Self Organizing Iterative
Process." Parts I, II \& III. {\it J. Theor. Biol.}, {\bf 178}  (1996): 255-312.


\bibitem{fenichel}
Fenichel, N., ``Persistence and smoothness of invariant manifolds for flows." {\it Ind. Univ. Math. J.},
{\bf 21}:3  (1971).

\bibitem{difftopo}
Guillemin, V.,  Pollack, A.,  {\it Differential topology}, Prentice-Hall, (1974). 

\bibitem{grew}
Grew, N., {\it Anatomy of Plants}, Rawlings, London (1682) (available on http://www.botanicus.org/)

\bibitem{hofmeister}
Hofmeister W., ``Allgemeine Morphologie der Gewachse." in {\it Handbuch der 
Physiologischen Botanik}, {\bf 1} Engelmann, Leipzig (1868):  405-664 

\bibitem{scottthesis}
Hotton, S., ``Symmetry of Plants", {\it Thesis}, University of California, Santa Cruz (1999)

\bibitem{jpgr} Hotton, S.,  Johnson, V.,  Wilbarger, J., Zwieniecki, K., Atela, P.,  Gol\'e, C. and Dumais, J. ``The Possible and the Actual in Phyllotaxis: Bridging the Gap between Empirical Observations and Iterative Models."  {\it J Plant Growth Regul}  {\bf 25} (2006):  313-323

\bibitem{auxinmjolsness} 
J\"onsson, H.,  Heisler, M.G.,  Shapiro, B.E.,  Meyerowitz, E.M., and Mjolsness, E. 
``An auxin-driven polarized transport model for phyllotaxis." 
{\it Proc. Natl. Acad. Sci.},  {\bf 103} (2006):  1633-1638


\bibitem{leelevitov}
Lee, H.W. \& Levitov L.S. ``Universality in Phyllotaxis: A Mechanical Theory."
{\it Symmetry in Plants,
} World Scientific (1998):  619-653  

\bibitem{HIVnature} Li, S., P. Hill, C.P.,  Sundquist, W.I. \& Finch, J.T. ``Image reconstructions of helical
assemblies of the HIV-1 CA protein."  {\it Nature}, {\bf 407} (2000): 409-413


\bibitem{koch} Koch A.J., Bernasconi, G.  and Rothen, F. ``Phyllotaxis as a Geometrical and Dynamical System." in {\it Symmetry in Plants}, World Scientific Publishers (1998):  459-486 

\bibitem{kunzthesis}
Kunz M. ``Phyllotaxie, billiard polygonaux et th\'eorie des nombres." {\it Th\`ese}, Universit\'e de Lausanne, Switzerland (1997)

\bibitem{munkres} Munkres, J. R. {\it Topology, a first course}, Prentice-Hall (1975)


\bibitem{auxinbern}  Reinhardt, D.,  Mandel, T and Kuhlemeier, C. 
``Physiologists Auxin Regulates the Initiation and Radial Position of Plant  Lateral Organs." 
{\it The Plant Cell}, {\bf 12} (2000):  507-518

\bibitem{robinson}Robinson, C.  {\it Dynamical Systems}, CRC Press, (1994)

\bibitem{shipman} Shipman PD, Newell AC. ``Polygonal planforms and phyllotaxis on plants."  {\it  J Theor Biol. } {\bf 236}(2) (2005):  154-97 



\bibitem{auxinprunsi}  Smith, R.S.,  Guyomarc'h, S.,  Mandel, T.,  Reinhardt, D.,  Kuhlemeier, C., and Prusinkiewicz, P. ``A plausible model of phyllotaxis."
{\it Proc. Natl. Acad. Sci.}, {\bf 103} (2006):  1301-1306 

\bibitem{snow}
Snow M. \& Snow R. ``Minimum Areas and Leaf Determination." {\it Proc. Roy.
Soc.}, {\bf B139} (1952): 545-566 

\bibitem{vaniterson}
Van Iterson G., {\it Mathematische und microscopisch-anatamische Studien
uber Blattstellungen, nebst Betraschungen uber der Schalenbau der
Miliolinen} Gustav-Fischer-Verlag, Jena (1907)

\bibitem{williams} Williams R. F., Brittain, E.G.  ``A geometrical model of phyllotaxis."  {\it Aust. J. Bot.} {\bf 32} (1984):  43-72 

\bibitem{weisse} Weisse, A.  ``Sketch of the mechanical hypothesis of leaf-position." In K. Goebel's {\it Organography of Plants}. I. Clarendon Press. Oxford. (1900):  74-84 


\bibitem{zagorska}
Zag\'orska--Marek  B.  ``Phyllotaxic Diversity in {\it Magnolia} Flowers." {\it Acta Soc. Bot. Poloniae} {\bf 63} (1994): 
117-137.
\end{thebibliography}
\end{document}